\documentclass[numsec,webpdf,modern,medium,namedate]{oup-authoring-template}

\onecolumn

\graphicspath{{Fig/}}

\theoremstyle{thmstyleone}%
\newtheorem{theorem}{Theorem}
\newtheorem{proposition}[theorem]{Proposition}%
\theoremstyle{thmstyletwo}%
\newtheorem{example}{Example}%
\newtheorem{remark}{Remark}%
\theoremstyle{thmstylethree}%
\newtheorem{definition}{Definition}
\usepackage[singlespacing]{setspace}
\usepackage{cancel}
\usepackage{comment}

\usepackage{tikz}
\usepackage[fontsize=9pt]{fontsize}

\usepackage{mathtools}
\usepackage{amsmath}
\usepackage{amssymb}
\usepackage{amsthm}

\newcommand{\erk}{\hfill \ensuremath{\Diamond}} 

\newcommand{\ci}{\mbox{\protect $\: \perp \hspace{-2.3ex}
\perp$ }}

\newcommand{\cd}{\,|\,}

\newtheorem{lemma}[theorem]{Lemma}

\newcommand{\ts}[1]{\textcolor{blue}{*** TS: #1 ***}} 

\begin{document}

\journaltitle{Journals of the Royal Statistical Society}
\DOI{DOI HERE}
\copyrightyear{XXXX}
\pubyear{XXXX}
\access{Advance Access Publication Date: Day Month Year}

\firstpage{1}


\title[Causal Equilibrium Systems]{Interpretable Causal Graphical Models for Equilibrium Systems with Confounding}

\author[$\ast$]{Kai Z. Teh}
\author[]{Kayvan Sadeghi}
\author[]{Terry Soo}

\authormark{Teh et al.}

\address[]{\orgdiv{Department of Statistical Science}, \orgname{University College London}, \orgaddress{\street{Gower Street}, \postcode{WC1E 6BT}, \state{London}, \country{UK}}}

\corresp[$\ast$]{Email for correspondence. \href{Email:email-id.com}{kai.teh.21@ucl.ac.uk}}



\abstract{
In applications,  quantities of interest are often modelled in equilibrium or an equilibrium solution 
is
sought.  
The presence
of confounding makes causal inference in this setting challenging. We provide interpretable graphical models for equilibrium systems with confounding using anterial graphs \citep{lkayvan}, a class of graphs containing directed acyclic graphs, ancestral graphs, 
and chain graphs. In this setting, we provide valid graphical representations of both counterfactual variables and observational variables, which we relate to counterfactual graphs \citep{cfgraph} and single-world intervention graphs \citep{swig}. 
As an application of this graphical representation, we provide an element-wise procedure of selecting adjustment sets that flexibly include and exclude given covariates.}


\maketitle

\section{Introduction}
Causal inference aims to predict the consequences of interventions in the form of causal effects. 
When interventional data from randomised control trials are unavailable, 
causal assumptions relating observational data to the intervened setting are required to estimate causal effects.

A prominent casual inference framework that captures these assumptions is the graphical model formulation by \cite{pearlbook}, 
based on directed acyclic graphs (DAGs) interpreted as structural causal models (SCMs). Edges in a DAG are all directed and thus are not suitable in capturing symmetric relations between variables such as relationships induced by 
equilibria 
and confounding. 
Ancestral graphs \citep{ancestral} and chain graphs \citep{chain} have been proposed to represent confounding \citep{ancestral, ancestralint} (see also \cite{Sad}) or variables  at 
equilibrium  \citep{chaingraphcausal}.

Anterial graphs \citep{lkayvan} provide a natural generalisation of chain graphs, ancestral graphs, and DAGs,
and also represent the marginalisation and conditioning of these classes of graphs \citep{margincond}. In addition to directed edges, anterial graphs also contain bidirected and undirected edges, which we will use to represent confounding and equilibrium relationships between variables, 
respectively. We will formalise these confounding and equilibrium relationships using a generalisation of SCMs.

Our anterial graphical model thus enables causal inference for equilibrium systems with confounding, and unifies and generalises causal graphical models described by \cite{ancestral, ancestralint} and \cite{chaingraphcausal}, for ancestral graphs and chain graphs, respectively. 
Causal inference using anterial graphs is also interpretable since causal notions such as interventions can be expressed purely graphically as local mechanistic manipulations on the intervened variables. To demonstrate this interpretability, in the same spirit as counterfactual graphs \citep{cfgraph} and single-world intervention graphs \citep{swig}, we will provide graphs that encode counterfactual conditional independence assumptions between observational variables and counterfactual variables in the case of anterial graphs. 

These counterfactual conditional independencies encode many important assumptions for methods in causal inference. One such assumption is the conditional exchangability assumption \citep{exch}, which is important in tasks such as selecting an adjustment set of confounder covariates to control for confounding between treatment and outcome variables. A variety of criteria have been proposed to select such a set of confounder covariates,  such as the backdoor criterion of  \citet[Chapter 3.3.1]{pearlbook} and  the pre-treatment heuristics of  \cite{confheuristics}. Recently, there has also been work addressing other aspects of such selection criteria, such as statistical efficiency \citep{effadj} and  the
limited structural knowledge of the underlying causal graph \citep{itergraphexp}.

Based on marginalisation and conditioning operations developed for anterial graphs \citep{margincond}, we provide an element-wise algorithm to select an adjustment set constrained to include and exclude given covariates, which can arise in practice due to prohibitive costs or regulation requirements in areas such as clinical trials or algorithmic fairness. Our algorithm is provably correct under conditions such as when the counterfactual graphs provided is valid. Compared to standard approaches of selecting adjustment sets where the inclusion and exclusion constraints are enforced post-hoc, after efficiently enumerating all possible adjustment sets \citep{construct}, our approach takes the constraints into account directly in the algorithm.

The structure of the paper will be as follows: Section \ref{bg} covers the relevant background, Section \ref{tnm1} covers the anterial graphical model of equilibrium systems with confounding, Section \ref{tnm2} covers interventions and the extension of counterfactual graphs in this setting using our anterial graphical model, and Section \ref{tnm3} covers the confounder selection algorithm. All proofs will be deferred to the final section. 

\section{Background}
\label{bg}
\subsection{Probabilistic Graphs}
Let \(\mathcal{G}\) denote a graph over a finite  set of nodes \(V\),
 with only one of the three types of edges: \emph{directed}  
(\(\rightarrow\)),   \emph{undirected} 
( --- ), and \emph{bidirected}
(\(\leftrightarrow\)), connecting 
two adjacent nodes. 
Consider a sequence \(\langle i_0,\ldots, i_n \rangle\) of nodes in the graph \(\mathcal{G}\).  If every pair of consecutive nodes is adjacent, then the sequence is a \emph{path};  if \(i_0=i_n\), in addition, then the sequence is a \emph{cycle}. If each edge between consecutive adjacent nodes \(i_m\) and \(i_{m+1}\) is either undirected or directed as \(i_m\xrightarrow{} i_{m+1}\), then the sequence is a \emph{semi-directed path} from \(i_0\) to \(i_n\); if in addition \(i_0=i_n\) and at least one of the edges is directed, then the sequence is a \emph{semi-directed cycle}.

%
In this work, we will consider \emph{anterial graphs} \citep{lkayvan}.
\begin{definition}[Anterial graphs]
\label{def: ant}
    An anterial graph 
    over a set of  nodes 
    is a graph which may contain directed,  undirected, 
    and bidirected edges that satisfies the following.
    
    \begin{enumerate}
        \item 
        There does not exist a semi-directed path between two nodes that are adjacent via a bidirected edge.
  %
        \item 
        The graph does not contain semi-directed cycles.
    \end{enumerate}
\end{definition}

By allowing only directed edges, we recover the definition of a \emph{directed acyclic graph} (DAG). Likewise excluding bidirected edges or node configurations of the form \(i-k\leftrightarrow j\) (which is trivially satisfied by excluding undirected edges), we recover the definition of chain graphs or ancestral graphs, respectively.

Following \cite{chain}, a \emph{chain component} \(\tau\) in an anterial graph \(\mathcal{G}\) is a maximal set of nodes such that every pair of nodes \(x,y\in \tau\) is connected by a sequence of undirected edges;
thus all nodes in \(\mathcal{G}\) can be partitioned into chain components.  Note that in the case of a DAG, each 
node 
is a chain component. 
Given a node \(i\in V\), let \(\tau(i)\) denote the chain component containing \(i\). Given a set of nodes \(C\subseteq V\) in a graph \(\mathcal{G}\), the \emph{neighbors} of \(C\), denoted by \(\text{ne}(C)\),  are the set of nodes \(x\not \in C\) such that \(x-i\) for some node \(i\in C\); the \emph{parents} of \(C\), denoted by \(\text{pa}(C)\),  are the set of nodes \(x\not \in C\) such that \(x\rightarrow i\) for some node \(i\in C\); the \emph{anterior} of \(C\), denoted by \(\text{ant}(C)\), is the set of nodes \(x\not \in C\) such that there exists a semi-directed path from \(x\) to some node \(i\in C\). 

\begin{remark}
When undirected edges are absent, the anterior of \(C\) coincides with the ancestors of \(C\). 
    Note that some authors let the neighbors and anterior of the set \(C\subseteq V\) contain the set \(C\) itself, however, 
    we do not. 
    \erk
\end{remark}

For disjoint subsets \(A,B,C \subseteq V\),  let \(A\perp_{\mathcal{G}} B\cd C\) denote the graphical separation of \(A\) and \(B\) given \(C\) in a graph \(\mathcal{G}\); in the case of anterial graphs, this is understood as the separation criterion in \citet[Section 3.2]{lkayvan}, which is simplified to the classical d-separation \citep[Chapter 1.2.3]{pearlbook} when the graph \(\mathcal{G}\) is a DAG, and m-separation for ancestral graphs; see also \citet[Section 3.3]{lkayvan}.  
For brevity, we will sometimes express the singleton set \(\{i\}\subseteq V\), without the set braces, as \(i\). The graphs $\mathcal{G}_1$ and 
$\mathcal{G}_2$ over the same node set \(V\) are 
\emph{Markov equivalent} if they induce the same graphical separations.


In a DAG, each pair of non-adjacent nodes are separated, say by their ancestors.   However,  an anterial graph need not be \emph{maximal} (see the graph \(\mathcal{G}^*_0\) in Figure \ref{exmax} for example).

\begin{definition}[Maximal graphs]\label{maxim}
    A graph \(\mathcal{G}\) is 
    maximal if for each pair of  nodes \(i\) and \(j\) in \(\mathcal{G}\), we have
    \begin{align*}
        i \text{ not adjacent to } j \text{ in \(\mathcal{G}\)} \Rightarrow i\perp_{\mathcal{G}} j\cd C \text{ for some \(C\subseteq V\backslash\{i,j\}\)}.
    \end{align*}
    \end{definition}

\cite{lkayvan} characterised when anterial graphs are maximal. They also provided a graphical operation \(\text{max}(\cdot)\) which \emph{maximises} any given anterial graph \(\mathcal{G}\)---returning a graph  \(\text{max}(\mathcal{G})\) with the following properties.
\begin{itemize}
        \item The graph \(\text{max}(\mathcal{G})\) is maximal.
        \item The graphs \(\text{max}(\mathcal{G})\) and \(\mathcal{G}\) are Markov equivalent.
        \item The graph \(\mathcal{G}\) is a subgraph of \(\text{max}(\mathcal{G})\).
\end{itemize}

We associate  a set of random variables \(X_V=(X_1,\ldots,X_{|V|})\) with a joint distribution \(P\)  to the set of nodes \(V\).  We will often  write $X \sim P$, if $P$ is the law of $X$. 
Given a set \(A\subseteq V\), let \(X_A =  (X_i)_{i\in A} \), and \(A\ci B\cd C\) denote the conditional independence of \( X_A \) and \( X_B \) given \( X_C \). 
If each graphical separation implies a corresponding conditional independence, then we have the  \emph{Markov property}: 
\begin{definition}
[Markov property]
    A distribution \(P\) is Markovian to \(\mathcal{G}\)  if \(A\perp_{\mathcal{G}} B\cd C \Rightarrow A\ci B\cd C\) for all disjoint \(A,B,C \subseteq V\).
\end{definition}

If we have the reverse implication as well, then we have \emph{faithfulness}. 

\begin{definition}[Faithfulness]
   A distribution   \(P\) is faithful to \(\mathcal{G}\) if   \(A\perp_{\mathcal{G}} B\cd C \iff A\ci B\cd C\) for all disjoint \(A,B,C \subseteq V\).
\end{definition}

If the distribution \(P\) is faithful to some anterial graph \(\mathcal{G}\), then \(P\) is a compositional graphoid \citep[Theorem 17]{kayvanfaith}.
\begin{definition}[Compositional graphoid]
    A distribution \(P\) over \(V\) is a compositional graphoid if we have the following conditional independence properties.  For disjoint \(A,B,C,D\subseteq V\),
    \begin{itemize}
    \item(Intersection)  \(A\ci B \cd C\cup D\) and \(A\ci D\cd C\cup B\) implies \(A\ci B\cup D\cd C\), and
    \item(Composition)  \(A\ci B\cd C\) and \(A\ci D\cd C\) implies \(A\ci B\cup D\cd C\). 
    \end{itemize}
\end{definition}
\begin{remark}
    Examples of compositional graphoids are given by multivariate Gaussian distributions. The assumption that a distribution is a compositional graphoid is often made when attempting to faithfully describe a distribution via a graphical model; see \cite{kayvanfaith}.  \hfill \ensuremath{\Diamond}
\end{remark}

Given a distribution \(P\) over variables \(X_V\), disjoint subsets \(M, C\subseteq V\) and \(x_B\), a fixed value of \(X_B\), let \(P_A(\cdot \cd x_B)\) denote the marginal of the conditional distribution \(P(\cdot \cd x_B)\) over \(X_A\) (by integrating out \(X_{V\backslash (A\cup B)}\)). 
For a distribution \(P_V\) that is faithful to a DAG \(\mathcal{G}\), there does not necessarily exist a DAG to which the marginal distribution \(P_{V\backslash M}\) is faithful. Likewise, there does not necessarily exist a DAG to which the conditional distribution \(P(\cdot \cd x_C)\) is faithful. 

In the case where \(P\) is faithful to an anterial graph \(\mathcal{G}\), anterial graphs to which the marginal or conditional distributions are faithful  always exist. Anterial graphs are thus closed under marginalising and conditioning, unlike DAGs, chain graphs and ancestral graphs. 
Given an anterial graph \(\mathcal{G}\) over the nodes \(V\) and subset \(M\subseteq V\), \cite{margincond} 
provides
the graphical operation \(\alpha_{\textnormal{m}}\) to marginalise \(\mathcal{G}\), which returns the 
anterial graph \(\alpha_{\textnormal{m}}(\mathcal{G};M)\) over nodes \(V\backslash M\) such that for disjoint \(A,B,C\subseteq V\backslash M\), we have
    \begin{align}\label{propmarg}
        A\perp_{\mathcal{G}} B\cd C \iff A\perp_{\alpha_{\textnormal{m}}(\mathcal{G};M)} B\cd C.
    \end{align}
This correspondence captures the relationship between the conditional independencies of a distribution \(P\) and its marginal distribution \(P_{V\backslash M}\).

Similarly with conditioning, given any anterial graph \(\mathcal{G}\) over the nodes \(V\) and a subset \(C\subseteq V\), \cite{margincond} 
provides the
graphical operation \(\alpha_{\textnormal{c}}\) to condition \(\mathcal{G}\), which returns 
the
anterial graph \(\alpha_{\textnormal{c}}(\mathcal{G};C)\) over nodes \(V\backslash C\) such that for disjoint \(A,B,D\subseteq V\backslash C\), we have
    \begin{align}\label{propcond}
        A\perp_{\mathcal{G}} B\cd D\cup C \iff A\perp_{\alpha_{\textnormal{c}}(\mathcal{G};C)} B\cd D.
    \end{align}
This correspondence captures the relationship between the conditional independencies of a distribution \(P\) and its conditional distribution \(P(\cdot \cd x_C)\).

We will be using these properties of marginalisation and conditioning to analyse the confounder selection algorithm proposed in Section \ref{tnm3}; see Theorem \ref{CCSworks} for Algorithm \ref{ccsalg}. 
\subsection{Interventions and Confounder Selection}

A DAG \(\mathcal{G}\) is often associated to a \emph{structural causal model} with mutually independent errors.

\begin{definition}[Structural causal model (SCM)]
    Let  \(\mathcal{G}\) be a DAG where every node \(i\) is associated with a function \(f_i\) and an 
    error random variable \(\epsilon_i\). 
    An 
    SCM
    consists of functional assignments recursively  defining observational random variables $X_V$ of the form
$X_i=f_i(X_{\textnormal{pa}(i)}, \epsilon_i),$
for every node \(i \in V\). 
    
\end{definition}

We let $P$ be the law of $X_V$, so that $P$ is the observational distribution. 
When the errors are mutually independent, the resulting joint distribution \(P\) 
is Markovian to the DAG \(\mathcal{G}\) \citep[Theorem 1.4.1]{pearlbook}; this result facilitates the use of graphical calculus on DAGs to manipulate the SCM for the purposes of causal inference,  
such as in the case of
adjustment set selection. 
We will prove an analogue  of this result  when the distribution \(P\) is induced via SCMs of equilibriums with confounding; see Theorem \ref{correctgenant}.  

Using an SCM, intervening on a treatment set \(C\subseteq V\) amounts to replacing the functional assignments of all \(X_i\) where \(i\in C\) with some fixed values $a_i$, while keeping the functional assignments of all \(X_j\) where \(j\not \in C\) fixed. 
This replacement results in the definition of a different set of intervened random variables \(X^{\text{do}(a_C)}_V=(X^{\text{do}(a_C)}_1, \ldots, X^{\text{do}(a_C)}_{|V|})\), with which we associate an interventional distribution \(P^{\text{do}(a_C)}\). 
The new variables \(X^{\text{do}(a_C)}_V\) will be referred to as  intervened variables in the scope of this work.
Note that \(X^{\text{do}(a_C)}_V\) are also known as potential outcomes 
\citep{poframework, Rubin01032005}.

Let \(O\subseteq V\) denote the outcome set and \(X_C\) the treatment variables.
Expressing common causal estimands such as the \emph{conditional average treatment effect} (CATE) solely in terms of the observational distribution depend on the following \emph{unconfoundedness} assumption: there exists an \emph{adjustment set} \(S \subseteq V\) such that
\begin{align}\label{defconfound}
X^{\text{do}(a_C)}_O\ci X_C\cd X_S 
\end{align} 
holds for all possible intervened values \(a_C\). If there is no subset of \(S\) that satisfies \eqref{defconfound}, then \(S\) is \emph{minimal}. The conditional independence \eqref{defconfound} is also known as the conditional exchangability condition \citep[Page 18]{exch}. 

Graphically, this coupling has been described using approaches such as twin-networks \citep{twin}, counterfactual graphs, \citep{cfgraph} and single-world intervention graphs \citep{swig, RichardsonRobins+2023}. These methods represent observational and intervened variables in a single graph \(\mathcal{G}'\) so that conditional independence relations between observational and intervened variables, such as \eqref{defconfound}, can be expressed as classical d-separation in \(\mathcal{G}'\).

Note that \eqref{defconfound} requires that the \emph{joint} distribution of the intervened outcome variable \(X^{\text{do}(a_C)}_O\) and the observational variable \(X_C\) to be well-defined;  
the required coupling between observational and intervened variables is specified by the shared errors \(\epsilon_V\) between the SCM and the intervened SCM. 

In this work, we will only be concerned with the kind of counterfactual assumptions that relate observational variables \(X_V\) and intervened variables \(X_V^{\textnormal{do}(a_C)}\) where the intervened value \(a_C\) is fixed, such as \eqref{defconfound}; relating intervened variables \(X_V^{\textnormal{do}(a_C)}\) and \(X_V^{\textnormal{do}(a'_C)}\) for different values \(a_C\) and \(a'_C\) is outside of our scope. Thus, we will consider only the set of treatment variables \(C\) and not the actual values \(a_C\) of the intervention, and this will be suppressed in our notation as \(X_V^{\text{do}(C)}\) and \(P^{\textnormal{do}(C)}\). Note that such assumptions are still sufficient for expressing some of the causal estimands involving different intervened values such as the \emph{effect of treatment on the treated} (ETT).

\subsection{Causal Interpretation of Chain Graphs}
Given a chain graph \(\mathcal{G}\), based on the factorisation property of chain graphs \citep{chain}, \cite{chaingraphcausal} has proposed causal interpretations of chain graphs as a data generating process having two nested processes. The outer process expresses the relationship between different chain components as
\begin{align}\label{basecase}
    X_\tau=f_\tau(X_{\textnormal{pa}(\tau)}, \epsilon_\tau)
\end{align}
for some random error variable \(\epsilon_\tau\) associated to each chain component \(\tau\), where the errors are jointly independent. These can be understood as SCM functional assignments by considering the chain component as nodes in a DAG. The inner process then associates each chain component \(\tau\) with a dynamical process that has a conditional equilibrium distribution given \(x_\textnormal{pa}(\tau)\).

\section{Structured Equilibrium Models with Confounding}
\label{tnm1}
We will formally express the relationship between different chain components described in \eqref{basecase} with potential confounding, as a structural equilibrium model.

\begin{definition}[Structural equilibrium model]\label{antscm}
Let \(\tau_1,\ldots, \tau_n\) be the parts of an  ordered partition of \(V\). Each part \(\tau_i\) is associated with a function \(f_{\tau_i}\) with arguments \(X_{\textnormal{pa}(\tau_i)}\), where \(\textnormal{pa}(\tau_i)\subseteq \{\tau_1,\ldots, \tau_{i-1}\}\), and 
an
error random variable \(\epsilon_{\tau_i}\). A structural equilibrium model consists of functional assignments of observational random variables $X_{\tau_i}$ of the following form
\begin{align*}
    X_{\tau_i}=f_{\tau_i}(X_{\text{pa}(\tau_i)}, \epsilon_{\tau_i}).
\end{align*}
\end{definition}

Note that the errors \(\epsilon_{\tau_1}, \ldots \epsilon_{\tau_n}\) need not be jointly independent. If all \(\tau_1,\ldots, \tau_n\) are all singletons with jointly independent errors, then  we recover the definition of an SCM for DAGs---in this case, \citet[Section 6.3] {chaingraphcausal} view Definition \ref{antscm} as a data generating process

The notation \(\textnormal{pa}(\tau_i)\) suggests a corresponding graph \(\mathcal{G}\). Indeed, given a structural equilibrium model, we construct the corresponding graph \(\mathcal{G}\) by considering, for each part \(\tau \in \{\tau_1, \ldots, \tau_n\}\), how \(\tau\) is connected to \(\textnormal{pa}(\tau)\) based on the functional assignment of \(\tau\). This construction is formalised as Algorithm \ref{obsalg}.

\begin{algorithm}[h]
   \caption{The corresponding graph \(\mathcal{G}\) of a structural equilibrium model}\label{obsalg}
   {\bfseries Input:} Structural equilibrium model
   
 {\bfseries Output:} Graph \(\mathcal{G}\).
\begin{algorithmic}[1]
\item[] For every part \(\tau\in \{\tau_1,\ldots, \tau_n\}\), let \(J^\tau(\cdot\cd x_{\textnormal{pa}(\tau)})\) be the law of \(f_{\tau}(x_{\text{pa}(\tau)}, \epsilon_\tau )\).
\item For nodes \(i\) and \(j\) in \(\tau\), if \(i\cancel{\ci} j\cd \tau\backslash \{i,j\}\) holds for \(J^\tau(\cdot\cd x_{\textnormal{pa}(\tau)})\), then add undirected edge \(i-j\).
\item For nodes \(i\in \textnormal{pa}(\tau)\) and \(j\in \tau\), if for all values of \(x_{(\tau\cup \textnormal{pa}(\tau))\backslash \{i,j\}}\), the conditional marginal \(J^{\tau}_i(\cdot\cd x_{\tau\backslash i};x_{\textnormal{pa}(\tau)})\) depends on the value of \(x_i\), then add the directed edge \(i\rightarrow j\).

    \item For parts \(\tau\) and \(\tau'\), if \(\epsilon_{\tau}\cancel{\ci} \epsilon_{\tau'}\), then add bidirected edges \(i\leftrightarrow j\) for every node \(i\in \tau\) and \(j\in \tau'\).
\end{algorithmic}
\end{algorithm}

For examples of structural equilibrium models and their corresponding graphs, see Figures \ref{simsetup1}, \ref{simsetup2}, and \ref{simsetup3} in Section \ref{sec:sim}.

\begin{remark}
Step 1 of Algorithm \ref{obsalg} ensures that \(J^{\tau}(\cdot\cd x_{\textnormal{pa}(\tau)})\) satisfies the pairwise Markov property with respect to (w.r.t.) to the constructed undirected graph over nodes in \(\tau\) 
and that no subgraph of the undirected graph satisfies this property.

The condition in Step 2 of Algorithm \ref{obsalg} can be succinctly expressed as  
    $i\cancel{\ci} j\cd \tau \cup \textnormal{pa}(\tau)\backslash\{i,j\}$
using the formulation in \citet[Theorem 3.5]{step2},  which allows the expression of dependence on fixed parameters \(x_{\textnormal{pa}(\tau)}\) as an extended conditional independence statement. These (extended) conditional independencies over \(\tau\) and \(\textnormal{pa}(\tau)\) are represented graphically using Algorithm \ref{obsalg}. See Figure \ref{proofill}.
\begin{figure}[h]
        \centering
        \includegraphics[width=0.2\linewidth]{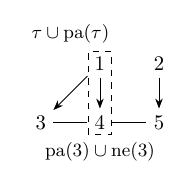}
        \caption{Chain component \(\tau = \{3,4,5\}\) and parents \(\textnormal{pa}(\tau) = \{1,2\}\). For every node \(i\in \tau\), Algorithm  \ref{obsalg} constructs \(\textnormal{ne}(i)\) and \(\textnormal{pa}(i)\) such that the conditional marginal \(J^{\tau}_i(\cdot \cd x_{\tau\backslash i}; x_{\textnormal{pa}(\tau)})\) does not depend on nodes outside of \(\textnormal{pa}(i)\cup \textnormal{ne}(i)\) (formalised as Lemma \ref{localind} in the Proofs Section).}
        \label{proofill}
    \end{figure}\erk


\end{remark}
For every part \(\tau\), Step 1 of Algorithm \ref{obsalg} constructs the chain components of the graph \(\mathcal{G}\) using the distribution \(J^\tau(\cdot\cd x_{\textnormal{pa}(\tau)})\) of \(f_{\tau}(x_{\text{pa}(\tau)}, \epsilon_{\tau})\) from the model. For every node \(i\in \tau\), step 2 of Algorithm \ref{obsalg} then uses the \(i\)-marginal of the conditional distribution \(J^{\tau}_i(\cdot\cd x_{\textnormal{ne} (i)};x_{\textnormal{pa}(\tau)})\) to determine the parents of \(i\) for the constructed \(\mathcal{G}\). Step 3 of the algorithm then includes bidirected edges in the constructed graph \(\mathcal{G}\) to represent dependent noises of the model.

\begin{remark}[Relating parts to chain components of \(\mathcal{G}\)]\label{multicompremakr}
     
     Consider the part \(\tau\) that decomposes in multiple connected components \(\tau'_1,\ldots, \tau'_m\)  after performing Step 1 of Algorithm \ref{obsalg}. These connected components are the chain components of the corresponding graph \(\mathcal{G}\). 
     
     Since Step 1 of Algorithm \ref{obsalg} constructs an undirected graph of chain components such that \(J^\tau(\cdot \cd x_{\textnormal{pa}(\tau)})\) satisfies the pairwise Markov property, if \(J^\tau(\cdot \cd x_{\textnormal{pa}(\tau)})\) satisfies the intersection property, then it can be seen that the chain components are jointly independent using the global Markov property, since the chain components are not connected by paths. 
    
    Thus, the functional assignment \(X_\tau=f_{\tau}(X_{\textnormal{pa}(\tau), \epsilon_\tau})\) for the part \(\tau\) can be
    decomposed into multiple functional assignments:
    \begin{align*}
        X_{\tau'_i}=f_{\tau'_i}(X_{\textnormal{pa}(\tau)}, \epsilon_{\tau'_i}),
    \end{align*}
    where \(\tau'_i\in \{\tau'_1,\ldots, \tau'_m\}\), with \(\epsilon_\tau=(\epsilon_{\tau'_1}, \ldots, \epsilon_{\tau'_n})\), where 
    the constituent errors
    are jointly independent.

    Having multiple chain components from a part does not affect the proof of Theorem \ref{correctgenant}. However, without loss of generality, we will consider structural equilibrium models such that, for each part, \(\tau\) does not decompose 
    w.r.t.\ the distribution \(J^\tau(\cdot \cd x_{\textnormal{pa}(\tau)})\). 
    Thus we will assume that each part \(\tau\) is a chain component in the corresponding graph \(\mathcal{G}\). \erk
\end{remark}

If the corresponding graph 
\(\mathcal{G}\) is anterial, then the distribution \(P\), assumed to be a compositional graphoid, induced from the structural equilibrium model is Markovian to \(\mathcal{G}\).

\begin{theorem}[Markov property for structural equilibrium models]\label{correctgenant}
    Let \(P\) be the joint distribution over the variables \(X_V\) induced from a structural equilibrium model and let the corresponding graph \(\mathcal{G}\) over the nodes \(V\), as constructed in Algorithm \ref{obsalg},  be anterial. If \(P\) is a compositional graphoid, then \(P\) is Markovian to \(\mathcal{G}\).
\end{theorem}

Throughout this work, we will assume structural equilibrium models that correspond to anterial graphs, via Algorithm \ref{obsalg}. In addition to Definition \ref{antscm}, these models have an additional constraint:  for any parts \(\tau'_1, \tau'_m\in \{\tau_1, \ldots, \tau_n\}\), if there is a sequence of parts \(\langle \tau'_1, \tau'_2, \ldots, \tau'_{m-1}, \tau'_{m}\rangle\) such that \( \tau'_i\) intersects \( \textnormal{pa}(\tau'_{i+1})\) for each consecutive parts \(\tau'_i\) and \(\tau'_{i+1}\), then \(\epsilon_{\tau'_1} \ci \epsilon_{\tau'_m}\). Indeed, this condition can be graphically expressed as the lack of semi-directed paths between endpoints of a bidirected edge. See Section \ref{sec:anterial} on why the anterial assumption is important for causal interpretability.

In the case where the corresponding graph is a chain graph,
\citet[Proposition 6] {chaingraphcausal} 
gives 
the Markov property for \(\mathcal{G}\) when viewing the chain graph as a DAG with chain components as nodes.  
Theorem \ref{correctgenant} enables the consideration of conditional independencies involving \emph{every} subset of nodes, addressing the remark from \cite{PDd} regarding the full extraction of conditional independencies from the chain graph.


When the parts of the structural equilibrium model are singletons, Algorithm \ref{obsalg} returns a graph 
of 
only directed and bidirected edges. When the errors in the structural equilibrium model are jointly independent, Algorithm \ref{obsalg} returns a chain graph. 
With the additional assumption that the induced distribution is a compositional graphoid, Theorem \ref{correctgenant} extends the Markov property results in \citet[Theorem 27]{Sad} for ancestral graphs and \cite{chaingraphcausal} for chain graphs.


Step 3 of Algorithm \ref{obsalg} connects either all or none of the nodes between the chain components \(\tau\) and \(\tau'\) with bidirected edges.

\begin{definition}[Chain-connected graph]
\label{defchain}
    A graph \(\mathcal{G}\) is \emph{chain-connected} if for all nodes 
    $i$ and $j$,
    we have
    \begin{align*}
        i\leftrightarrow j \Rightarrow i\leftrightarrow k \quad \text{ for all }k\in \tau(j).
    \end{align*}
\end{definition}

Thus, it can be seen that the corresponding graph \(\mathcal{G}\) is chain-connected. This can be interpreted as all nodes in the chain component \(\tau\) sharing the same error \(\epsilon_\tau\).  Throughout this work, we will describe structural equilibrium models using chain-connected anterial graphs.  Common graph classes such as DAGs, ancestral graphs and chain graphs are still a subclass of chain-connected anterial graphs.

\begin{remark}
    In principle, Step 3 of Algorithm \ref{obsalg} can be modified to add bidirected edges based on conditional independencies of the induced joint distribution \(P\), while still preserving the Markov property required. However, this approach is instead distributional, depending on the actual coupling of the errors, and no longer purely structural. depending only on the independencies of the errors, which allows for causal operations in an interpretable graphical manner. See also a similar argument on the importance of the anterial assumption on causal interpretability in Section \ref{sec:anterial}. \erk
\end{remark}

\subsection{Interpreting the constructed graphical model \(\mathcal{G}\)}

In \cite{chaingraphcausal}, each chain component \(\tau\) of a chain graph model is associated with a dynamical process, such as a Gibbs sampler \citep{gibbs}, having a conditional equilibrium distribution given \(x_{\textnormal{pa}(\tau)}\), which we denote as \(J^\tau(\cdot\cd x_{\textnormal{pa}(\tau)})\). As suggested by the notation, given a structural equilibrium model, for every part \(\tau\), we will similarly associate the law \(J^\tau(\cdot\cd x_{\textnormal{pa}(\tau)})\) of 
\(f_\tau(x_{\textnormal{pa}(\tau)}, \epsilon_\tau)\) 
with the equilibrium distribution 
of the Gibbs sampler as follows. 

Label the nodes in the chain component from \(1\) to \(|\tau|\). Let \(X^0_\tau=x^0_\tau\) be some initial values and \(m_0=1\) be the starting node of the Gibbs sampler, at every timestep \(t\) the values \(X^t_\tau\) and node \(m_t\) are updated as follows:

\begin{align}
\label{dyna1}
        m_t&=m_{t-1}+1 \text{ mod } |\tau|, \nonumber\\
        X_{\tau\backslash m_t}^t&= x_{\tau\backslash m_t}^{t-1}, \text{ and }\nonumber\\
        X_{m_t}^t&\sim J^{\tau}_{m_t}(\cdot \cd x_{\text{ne}(m_t)}^{t-1};  x_{\text{pa}(m_t)}).
    \end{align}

  \begin{remark}
    Given a structural equilibrium model, interpreting \(J^\tau(\cdot\cd x_{\text{pa}(\tau)})\) as the equilibrium distribution of some dynamical process is only required to define sensible notions of interventions  in Section \ref{tnm2}. Observationally, it suffices to consider \(J^\tau(\cdot\cd x_{\text{pa}(\tau)})\) only, since Theorem \ref{correctgenant} is a statement on the observational joint distribution and therefore holds regardless of the exact dynamical process. \erk
\end{remark}

Despite anterial graphs being closed under marginalisation and conditioning, 
we will discuss why interpretations based on marginalisation and conditioning of existing graph classes are not suitable for chain-connected anterial graphs.

Marginalising chain graphs results in anterial graphs \citep{margincond}. Building on  top of existing chain graph interpretations \citep{chaingraphcausal}, one may be tempted to consider anterial graphs as the marginal of a chain graph with latent variables, with bidirected edges \(i\leftrightarrow j\) representing a latent variable between nodes \(i\) and \(j\). However, Example \ref{antmarg} shows the problem with such an interpretation.


\begin{example}[Problem with interpretations based on marginalisation]\label{antmarg}
 \begin{figure}[h]
\centerline{\includegraphics[width=0.7\linewidth]{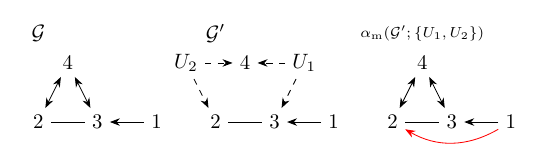}}
\caption{Left: Chain-connected anterial graph \(\mathcal{G}\). Middle: Interpretation of bidirected edges \(3\leftrightarrow 4\) and \(2\leftrightarrow 4\) in \(\mathcal{G}\) as unobserved latent variables \(U_1\) and \(U_2\) respectively, from the chain graph \(\mathcal{G}'\). Right: Graph \(\alpha_{\textnormal{m}}(\mathcal{G}';\{U_1,U_2\})\) obtained by marginalising \(\mathcal{G}'\) over the latent variables \(U_1\) and \(U_2\).}
\label{failedmarg}
\end{figure}
 
Consider Figure \ref{failedmarg} and interpret the bidirected edges in the chain-connected anterial graph \(\mathcal{G}\) as the existence of unobserved latent variables \(U_1\) and \(U_2\), as shown in the chain graph \(\mathcal{G}'\).  Marginalising over $U_1$ and $U_2$ in \(\mathcal{G}'\) results in \(\alpha_{\textnormal{m}}(\mathcal{G}'; \{U_1,U_2\})\) which is not Markov equivalent to \(\mathcal{G}\); indeed, 
\(1\perp_{\mathcal{G}} 2\cd 3\), 
but 
this separation does not hold in the marginalisation \(\alpha_{\textnormal{m}}(\mathcal{G}'; \{U_1, U_2\})\).
\end{example}

Similar issues with marginal chain graph models have been raised 
by \cite{seggraph}
 and were resolved using graphs containing undirected, directed, and bidirected edges. However, the directed edges in such graphs do not seem to have a clear interpretation \citep[Section 2]{seggraph}.  

Similarly,  conditioning ancestral graphs results in anterial graphs \citep{margincond}. Building on top of existing interpretations of ancestral graphs \citep{ancestral, Sad}, it may also be  tempting to interpret anterial graphs as the conditional of ancestral graphs with selection variables, where the undirected edges \(i-j\) is interpreted as conditioning on a selection bias variable between the  nodes \(i\) and \(j\) from an ancestral graph. However, Example \ref{antcond} shows a similar problem with such an interpretation.
  \begin{figure}[h]
\centerline{\includegraphics[width=0.7\linewidth]{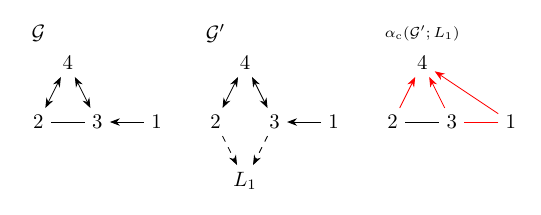}}
\caption{Left: Chain-connected anterial graph \(\mathcal{G}\). Middle: Interpretation of the undirected edge \(2-3\) as the selection bias variable \(L_1\) from the ancestral graph \(\mathcal{G}'\). Right: Graph \(\alpha_{\textnormal{c}}(\mathcal{G}';L_1)\) obtained by conditioning \(\mathcal{G}'\) over the selection bias variable \(L_1\).}
\label{failedcond}
\end{figure}

\begin{example}[Problem with interpretations based on conditioning]\label{antcond}
Consider Figure \ref{failedcond} and interpret the undirected edge \(2-3\) in the chain-connected anterial graph \(\mathcal{G}\) as the existence of selection variables \(L_1\), as shown in the ancestral graph \(\mathcal{G}'\). Conditioning on \(L_1\) in \(\mathcal{G}'\) results in \(\alpha_{\textnormal{c}}(\mathcal{G}'; L_1)\), which is not Markov equivalent to \(\mathcal{G}\); indeed, \(1\perp_\mathcal{G} 4 \) but this separation does not hold in the conditional graph \(\alpha_{\textnormal{c}}(\mathcal{G}'; L_1)\).
\end{example}

\subsection{Examples and simulations}\label{sec:sim}

To illustrate Theorem \ref{correctgenant}, we will simulate data from structural equilibrium models that correspond to the chain-connected anterial graphs in Figures \ref{simsetup1}, \ref{simsetup2}, and \ref{simsetup3} in Python. Since the joint distribution in Theorem \ref{correctgenant} is induced from the distributions over chain components \(J^\tau(\cdot \cd x_{\textnormal{pa}(\tau)})\), which we interpret as equilibrium of Gibbs samplers in \eqref{dyna1}, the simulation also allows us to test the validity of the Markov property of the structural equilibrium model when interpreted as Gibbs samplers in finite-time. 

 \begin{figure}[h]
\centerline{\includegraphics[width=\linewidth]{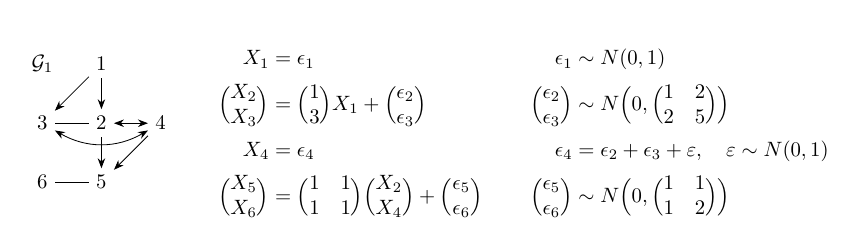}}
\caption{Left: Corresponding chain-connected anterial graph \(\mathcal{G}_1\). Middle: Structural equilibrium model. Conditioning on \(5\), the distribution \(J^{\{5,6\}}_6(\cdot \cd x_{5};x_{\{2,4\}})\) does not depend on \(x_{\{2,4\}}\), thus \(\textnormal{pa}(6)=\emptyset\). Right: Distribution and coupling of the error variables.}
\label{simsetup1}
\end{figure}

We will consider the structural equilibrium model which corresponds to the graph \(\mathcal{G}_1\) in Figure \ref{simsetup1}. We sample 10000 data points directly from the structural equilibrium model (using the equilibrium distributions) and 1000 data points using a Gibbs sampler of the conditional normal distributions \(J^{\{2,3\}}(\cdot \cd x_{1})\) and \(J^{\{5,6\}}(\cdot \cd x_2, x_4)\) with a burn-in of 10000 steps.

Using both sampled datasets, we perform a Fisher-Z test (using the \verb|causal-learn| \citep{python} package) with a null hypothesis of pairwise conditional independence of the form  \(i\ci j \cd \textnormal{ant}(i,j)\) for every pair of nodes $i$ and $j$. The Markov property w.r.t.\ the  graph \(\mathcal{G}_1\) would imply pairwise conditional independencies for some nodes \(i\) and \(j\). The p-values for every pair of nodes are shown in Figure \ref{g1res} in Section \ref{sec:simres}. 

Next, we consider the structural equilibrium model which corresponds to the graph \(\mathcal{G}_2\) in Figure \ref{simsetup2}. Again, the p-values for the pairwise conditional independence tests are shown in Figure \ref{g2res} in Section \ref{sec:simres}.

 \begin{figure}[h]
\centerline{\includegraphics[width=\linewidth]{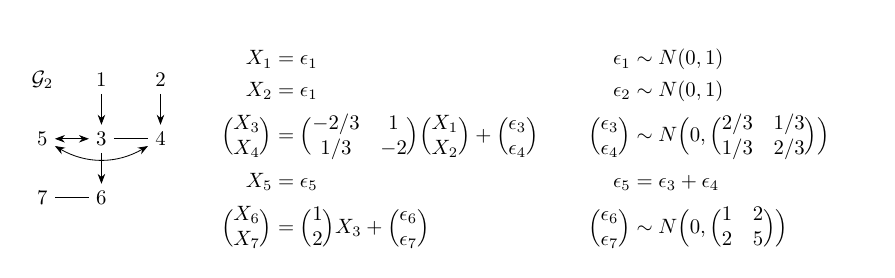}}
\caption{Left: Corresponding chain-connected anterial graph \(\mathcal{G}_2\). Middle: Structural equilibrium model. Conditioning on \(3\), the distribution \(J^{\{3,4\}}_4(\cdot \cd x_{3};x_{\{1,2\}})\) depends on \(x_{2}\) only, thus \(\textnormal{pa}(4)=2\), similarly we have \(\textnormal{pa}(3)=1\). Right: Distribution and coupling of the error variables.}
\label{simsetup2}
\end{figure}

We will use the structural equilibrium model which corresponds to the graph \(\mathcal{G}_3\) in Figure \ref{simsetup3}. 

 \begin{figure}[h]
\centerline{\includegraphics[width=\linewidth]{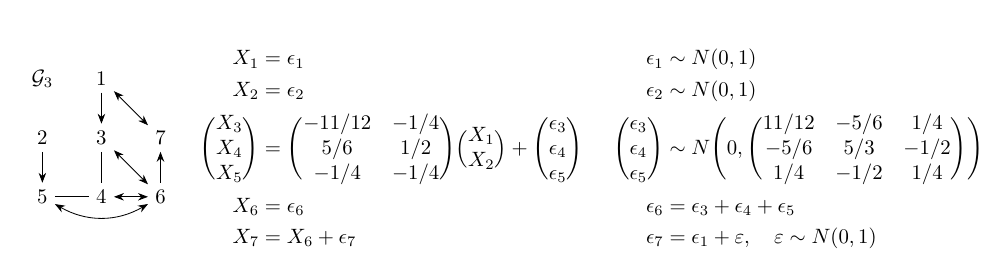}}
\caption{Left: Corresponding Chain-connected anterial graph \(\mathcal{G}_1\). Middle: Structural equilibrium model. \(3\ci 5\) holds for the distribution \(J^{\{3,4,5\}}(\cdot \cd x_{\{2,4\}})\), thus \(3\) is not connected to \(5\). Conditioning on \(\textnormal{ne}(4)\), the distribution \(J^{\{3,4,5\}}_4(\cdot \cd x_{\{3,5\}};x_{\{1,2\}})\) does not depend on \(x_{\{1,2\}}\), thus \(\textnormal{pa}(4)=\emptyset\). Likewise for \(\textnormal{pa}(3)=1\) and \(\textnormal{pa}(5)=2\). Right: Distribution and coupling of the error variables.}
\label{simsetup3}
\end{figure}

The p-values for the pairwise independence tests are shown in Figure \ref{g3res} in the Section \ref{sec:simres}.  

In Figure \ref{g3res}, high p-values indicate insufficient evidence for rejecting the null hypothesis in favour of the alternative hypothesis where \(i\cancel{\ci} j \cd \textnormal{ant}(i,j)\), implying conditional independence. We see that the Fisher Z-test rejects the null hypothesis in a manner consistent with the conditional independencies implied by the Markov property of Theorem \ref{correctgenant}. Note that due to the higher dimensionality when sampling \(\{3,4,5\}\), the p-values of the pairwise conditional independencies \(5\ci 1 \cd \{2,3,4\}\) and \(5\ci 3\cd \{2,1,4\}\) are on the order of about \(10^{-7}\), while all the other returned p-values are close to \(0\).

\section{Graphical Interventions of Structural Equilibrium Models}\label{tnm2}

Given a structural equilibrium model with the corresponding graph \(\mathcal{G}\), using the interpretation of a chain component \(\tau\) of \(\mathcal{G}\) as the Gibbs sampler in \eqref{dyna1}, an intervention on a treatment set \(C\subseteq V\) (to intervened values \(a_C\)) can be expressed through the following intervened Gibbs sampler by replacing the updating of the treatment nodes in \(C\) with fixed assignments. This replacement results in the \emph{intervened} Gibbs sampler as follows, expressed using the notation of the Gibbs sampler in \eqref{dyna1}:

\begin{align}\label{intervenescm}
        m_t&=m_{t-1}+1 \text{ mod } |\tau|, \nonumber\\
        \text{if } m_t\in C, & \nonumber\\
          X_{\tau\backslash C}^t&= x_{\tau\backslash C}^{t-1}, \textnormal{ and }\nonumber\\
        X_{\tau\cap C}^t&=a_{\tau\cap C}, \nonumber\\
       \text{if } m_t\not\in C, &\nonumber\\
       X_{\tau\backslash m_t}^t&= x_{\tau\backslash m_t}^{t-1}, \textnormal{ and }\nonumber\\
         X_{m_t}^t&\sim J^\tau_{m_t}(\cdot \cd x_{\text{ne}(m_t)}^{t-1};  x_{\text{pa}(m_t)}),
    \end{align}
 with the conditional distribution \(J^\tau_{m_t}(\cdot \cd x_{\text{ne}(m_t)}^{t-1};  x_{\text{pa}(m_t)}) \)  being the same in the Gibbs sampler \(\eqref{dyna1}\) of the structural equilibrium model.

The resulting equilibrium distribution \(J^{\tau,\text{do}(C)}(\cdot \cd x_{\text{pa}(\tau)})\) on the chain component \(\tau\) is 
\begin{align}\label{keyeqint}
    J^{\tau,\text{do}(C)}(\cdot \cd x_{\text{pa}(\tau)})=J^\tau_{\tau\backslash C}(\cdot \cd a_{\tau\cap C}; x_{\textnormal{pa}(\tau)})\delta_{\tau\cap C}(a_{\tau\cap C}),
\end{align} 
where \(J^\tau_{\tau\backslash C}(\cdot\cd a_{\tau\cap C};x_{\text{pa}(\tau)})\) is the conditional distribution of the equilibrium distribution \(J^\tau(\cdot \cd x_{\text{pa}(\tau)})\) of the pre-intervened Gibbs sampler \eqref{dyna1} and \(\delta_{\tau\cap C}(a_{\tau\cap C})\) is the delta function that takes the value \(1\) when \(X_{\tau\cap C}\) takes the intervened values \(a_{\tau \cap C}\); this is a local  version of \citet[Section 6.4, Equation 18]{chaingraphcausal}.


From the distributions \(J^{\tau,\text{do}(C)}(\cdot \cd x_{\text{pa}(\tau)})\), we 
define the resulting \emph{intervened structural equilibrium model} by replacing the assignments for all chain components \(\tau\) such that \(\tau\cap C\neq \emptyset\) in the structural equilibrium model with
\begin{align}\label{intscm}
    X_\tau=f^{\text{do}(C)}_\tau(X_{\text{pa}(\tau)}, \epsilon_\tau),
\end{align}
such that
\begin{itemize}
    \item \(f^{\text{do}(C)}_\tau(x_{\text{pa}(\tau)}, \epsilon_\tau)\sim J^{\tau,\text{do}(C)}(\cdot \cd x_{\text{pa}(\tau)})\), and 
    \item the chain component errors \(\epsilon_{\tau_1}, \ldots, \epsilon_{\tau_n}\) have the same joint distribution as the errors in the observational structural equilibrium model (pre-intervention).
\end{itemize}

The intervened structural equilibrium model induces a new set of intervened random variables \(X^{\text{do}(C)}_V\) with which we associate a joint interventional distribution \(P^{\text{do}(C)}\). 

\begin{remark}
    Generally, in structural equilibrium models, the joint interventional distribution is not equivalent to the joint conditional distribution. However, this equivalence holds for the conditional distribution of the variable given its parents. This equivalence is compatible with the relationship in \eqref{keyeqint}. \erk
\end{remark}

\begin{remark}
\cite{chaingraphcausal} also consider Langevin diffusions in lieu of Gibbs samplers.   Given an observational distribution \(P\), the interventional distribution \(P^{\textnormal{do}(c)}\) obtained from an intervened dynamical process using \eqref{intscm}  depends on the dynamical process being considered.

A Gibbs sampler is chosen for this work as the dynamical process, since the resulting intervention can be expressed purely graphically (see Algorithm \ref{intalg}); the same cannot be said about other dynamical processes in \cite{chaingraphcausal} such as Langevin diffusions, for which the notion of intervention defined will depend on diffusion parameters. 

     Note that without assuming the dynamical process to be a Gibbs sampler, an equilibrium distribution may not exist after intervening the dynamical process. \erk
 \end{remark}
 
\subsection{Representing interventions graphically}

Via a graphical operation (Algorithm \ref{intalg}) on the corresponding graph \(\mathcal{G}\) of the \emph{observational} structural equilibrium model, we can describe the interventional distribution \(P^{\textnormal{do}(C)}\) after intervening on a treatment set \(C\).

\begin{algorithm}[h]
   \caption{Intervening chain-connected anterial graphs}\label{intalg}
   {\bfseries Input:} chain-connected anterial graph $\mathcal{G}$, node subset $C\subseteq V$
   
 {\bfseries Output:} chain-connected anterial graph \(\text{do}_C(\mathcal{G})\)
\begin{algorithmic}[1]
\item[] For all \(i\in C\),
\item for nodes \(j\) such that \(i-j\), if \(j\not\in C\), then replace with \(i\rightarrow j\); if \(j\in C\), then remove edge \(i-j\), 
    \State remove all directed edges into \(i\), and
    \item remove all bidirected edges \(i\leftrightarrow j\).
\end{algorithmic}
\end{algorithm}

\begin{figure}[h]
\begin{center}
\centerline{\includegraphics[width=0.5\linewidth]{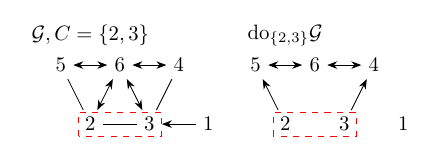}}
\caption{Left: Input chain-connected anterial graph \(\mathcal{\mathcal{G}}\) and treatment set \(C=\{2,3\}\). Right: Output \(\textnormal{do}_C(\mathcal{G})\) after removing bidirected and directed edges into \(C\).}
\label{intalgex}
\end{center}
\end{figure}

Algorithm \ref{intalg} first modifies undirected edges of nodes in the treatment set \(C\) and then removes directed and bidirected edges connected to \(C\). See Figure \ref{intalgex} for an illustration of Algorithm \ref{intalg}.

When \(\mathcal{G}\) is a DAG, only Step 2 of Algorithm \ref{intalg} is performed, and Algorithm \ref{intalg} reduces to a standard graphical intervention for DAGs. Since interventions can be expressed distributionally as a truncated factorisation of the observational distribution, the interventional distribution is Markovian to \(\textnormal{do}_C(\mathcal{G})\) \citep[Theorem 3.6]{spirtesbook}.

Theorem \ref{intmarkov} similarly describes \(P^{\textnormal{do}(C)}\) from structural equilibrium models using the intervened (anterial) graph \(\textnormal{do}_C({\mathcal{G}})\).

\begin{theorem}[Markov property for interventions]\label{intmarkov}
    Let \(P^{\textnormal{do}(C)}\) be the interventional distribution (from intervening on a treatment set \(C\)) induced from a structural equilibrium model corresponding to  the chain-connected anterial graph \(\mathcal{G}\).
    If \(P^{\textnormal{do}(C)}\) is a compositional graphoid, then \(P^{\textnormal{do}(C)}\) is Markovian to \(\textnormal{do}_C(\mathcal{G})\).
\end{theorem}

Note that the intervened structural equilibrium model is still a structural equilibrium model. Using the relationship in \eqref{keyeqint}, Algorithm \ref{intalg} constructs the corresponding graph \(\textnormal{do}_C(\mathcal{G})\) 
from the observational graph \(\mathcal{G}\). 
Theorem \ref{correctgenant} is  used to show Theorem \ref{intmarkov}.

\begin{remark}[The splitting of chain components when intervening can be dealt with similarly as in Remark \ref{multicompremakr}.]
 In Figure \ref{intalgex}, the chain component \(\{2,3,4,5\}\) in \(\mathcal{G}\) is split into multiple chain components in \(\textnormal{do}_C(\mathcal{G})\).
    In general, given a chain component \(\tau\) in \(\mathcal{G}\), the nodes in \(\tau\backslash C\) can be partitioned into different chain components \(\tau_i^\textnormal{do}\) in  \(\textnormal{do}_C(\mathcal{G})\), such that \(\bigcup_{i=1}^n \tau^{\text{do}}_{i}=\tau\backslash C\).

As argued in Remark \ref{multicompremakr} (by the global Markov property of the constructed \(\tau\)), these components \(\tau^\textnormal{do}_i\) are jointly independent given \(a_{\tau\cap C}\), then we can decompose the functional assignment \(X_\tau=f^{\text{do}(C)}_\tau(X_{\text{pa}(\tau)}, \epsilon_\tau)\) in \eqref{intscm} into
\begin{align*}
    X_{\tau\cap C}&=a_{\tau\cap C}, \textnormal{ and }\nonumber\\
    X_{\tau^{\textnormal{do}}_i}&=f_{\tau^{\textnormal{do}}_i}(X_{\textnormal{pa}(\tau)}, X_{ \tau\cap C}, \epsilon_{\tau^{\textnormal{do}}_i}),
\end{align*}
for each \(\tau^{\textnormal{do}}_i\in \{\tau^{\textnormal{do}}_1,\ldots, \tau^{\textnormal{do}}_n\}\),  such that \(\epsilon_\tau=(\epsilon_{\tau^{\textnormal{do}}_1}, \ldots, \epsilon_{\tau^{\textnormal{do}}_n})\) where each \(\epsilon_{\tau^{\textnormal{do}}_i}\) are jointly independent and \(f_{\tau^{\textnormal{do}}_i}(x_{\textnormal{pa}(\tau)}, x_{\tau\cap C}, \epsilon_{\tau^{\textnormal{do}}_i})\sim J^\tau_{\tau^{\textnormal{do}}_i}(\cdot\cd a_{\tau\cap C};x_{\text{pa}(\tau)})\). \erk
\end{remark}
\subsection{Counterfactual Interpretations}\label{sec:counterfactual}

In line with counterfactual interpretations in \cite[Chapter 7]{pearlbook}; \cite{twin}, after intervening on a treatment set \(C\), using the structural equilibrium model, a coupling between observational variables \(X_V\) and intervened variables \(X^{\text{do}(C)}_{V}\) is induced via common errors. Here, we describe this coupling, and provide a graphical representation of this coupling.

Consider a structural equilibrium model with the corresponding chain-connected anterial graph \(\mathcal{G}\). Given some chain component \(\tau\) in \(\mathcal{G}\), intervention on a treatment set \(C\subseteq V\) by setting \(X_C=a_C\) for some fixed value \(a_C\) induces a coupling between observational variables \(X_{\tau}\) and intervened variables \(X^{\text{do}(C)}_{\tau}\) via 
the
common error \(\epsilon_\tau\) of the functional assignments of the combined structural equilibrium model, shown as follows:

\begin{align}\label{swigSCM}
 X_\tau&=f_\tau(X_{\text{pa}(\tau)}, \epsilon_\tau), \textnormal{ and }\nonumber\\
    X^{\text{do}(C)}_\tau&=f^{\text{do}(C)}_\tau(X^{\text{do}(C)}_{\text{pa}(\tau)}, \epsilon_\tau), \quad  \text{if } \tau\cap C\neq \emptyset\nonumber\\
        X^{\text{do}(C)}_\tau&=f_\tau(X^{\text{do}(C)}_{\text{pa}(\tau)}, \epsilon_\tau),  \quad \text{if } \tau\cap C=\emptyset,
    \end{align}
where \(f^{\text{do}(C)}_\tau\) is the intervened function from \eqref{intscm}. 

\subsubsection{Representing the coupling graphically}

Given an anterial graph \(\mathcal{G}\), and a node subset \(C\subseteq V\), we apply the graphical operation introduced in Algorithm \ref{newswagalg}, and denote the resulting graph as \(\phi(\mathcal{G};C)\). See Figure \ref{newswagfig} for an illustration of Algorithm \ref{newswagalg}. 

\begin{algorithm}[h]
   \caption{Counterfactual Graph, \(\phi\)}
   \label{newswagalg}
   {\bfseries Input:} anterial graph $\mathcal{G}$, node subset $C\subseteq V$
   
 {\bfseries Output:} graph \(\phi(\mathcal{G};C)\)
\begin{algorithmic}[1]
   \State Reproduce \(\mathcal{G}\) and construct \(\textnormal{do}_\textnormal{C}(\mathcal{G})\), relabelling all the nodes \(i\) in \(\textnormal{do}_\textnormal{C}(\mathcal{G})\) as \(i^{\textnormal{do}(C)}\).
   
   \State For nodes \(i\not\in C\) in \(\mathcal{G}\) such that \(C\cap \textnormal{ant(i)}=\emptyset\), remove \(i^{\textnormal{do}(C)}\) and replace edges between \(i^{\textnormal{do}(C)}\) and nodes of the form \(j^{\textnormal{do}(C)}\) with the same edge between \(i\) and \(j^{\textnormal{do}(C)}\).
   
   \State Add bidirected edges based on Table \ref{easyalg}.
\end{algorithmic}
\end{algorithm}
\begin{table}[h]
    \centering
    \begin{tabular}{|c l|l|}
    \hline
     & If in \(\mathcal{G}\)  &  Then \\
         \hline
         1.&\(i\not \in C\) & add \(i^{\text{do}(C)}\leftrightarrow i\) \\
     2. &\(i\leftrightarrow j\) and \(i\not \in C\) & add \(i^{\text{do}(C)}\leftrightarrow j\)\\
     3. &\(j\in \tau(i)\) and \(i\not \in C\) &  add \(i^{\text{do}(C)}\leftrightarrow j\)\\
      
         \hline
    \end{tabular}
    \caption{Bidirected edges between relabelled nodes in \(\textnormal{do}_C(\mathcal{G})\) and nodes in \(\mathcal{G}\), based on conditions on the corresponding nodes in \(\mathcal{G}\). }
    \label{easyalg}
\end{table}
\begin{figure}[h]
\begin{center}
\centerline{\includegraphics[width=\columnwidth]{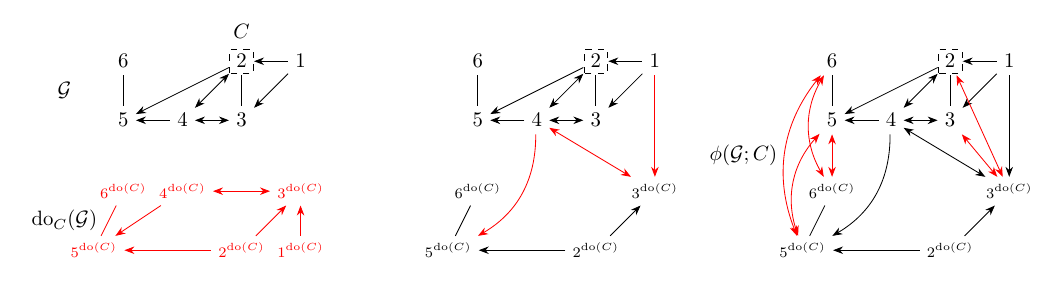}}
\caption{Consider the input anterial graph \(\mathcal{G}\) with \(C=2\). Left: Algorithm \ref{newswagalg} first creates \(\textnormal{do}_C(\mathcal{G})\) using Algorithm \ref{intalg} with the nodes relabelled alongside input graph \(\mathcal{G}\). Middle: Algorithm \ref{newswagalg} merges nodes in \(\mathcal{G}\) and \(\textnormal{do}_C(\mathcal{G})\). Right: Algorithm \ref{newswagalg} then introduces bidirected edges between \(\mathcal{G}\) and \(\textnormal{do}_C(\mathcal{G})\) based on Table \ref{easyalg} to obtain \(\phi(\mathcal{G}; 2)\).}
\label{newswagfig}
\end{center}
\end{figure}

Algorithm \ref{newswagalg} first introduces all the intervened variables \(X^{\textnormal{do}(C)}\) alongside the observational variables \(X_V\). For a node \(i\) in \(\mathcal{G}\) such that there does not exist a semi-directed path in \(\mathcal{G}\) from treatment set 
\(C\) to node \(i\), the observational variable \(X_i\) is  equal to \(X_i^{\textnormal{do}(C)}\). Step 2 of Algorithm \ref{newswagalg} accounts for this equality by merging all such observational variables \(X_i\) and intervened variables \(X^{\textnormal{do}(C)}_i\), represented by the nodes \(i\) and 
$i^{\textnormal{do}(C)}$, respectively, into the node \(i\).
The edges between \(i^{\textnormal{do}(C)}\) and the other nodes \(j^{\textnormal{do}(C)}\) are preserved as the same edge between \(i\) and \(j^{\textnormal{do}(C)}\).

Table \ref{easyalg} 
introduces bidirected edges between nodes \(i^{\textnormal{do}(C)}\) and nodes \(j\) in \(\mathcal{G}\) to capture the dependence of the chain component errors \(\epsilon_{\tau(i)}\) and \(\epsilon_{\tau(j)}\) that are shared between the corresponding variables \(X_{\tau(i)}^{\textnormal{do}(C)}\) and \(X_j\) (since errors are shared in \eqref{swigSCM}), when the corresponding node \(i\) in \(\mathcal{G}\) is not in treatment set \(C\).

In the case of DAGs, Algorithm \ref{newswagalg} reduces to the construction of counterfactual graphs from \cite{cfgraph}, where the common errors are marginalised over.  


Using the structural equilibrium model (corresponding to graph \(\mathcal{G}\)), the coupling between observational and intervened variables is described via \(\phi(\mathcal{G};C)\). Let node \(i^{\text{do}(C)}\) in graph \(\phi(\mathcal{G};C)\) represent \(X^{\text{do}(C)}_i\).

\begin{theorem}[Markov property of the counterfactual interpretation]\label{correctnewswag}

Let the observational random variables \(X_V\) be induced from a structural equilibrium model corresponding to the chain-connected anterial graph \(\mathcal{G}\). Let the combined structural equilibrium model in \eqref{swigSCM} induce the joint distribution \(P^*\) be the joint distribution over \(X_V\) and \(X_V^{\textnormal{do}(C)}\). If \(P^*\) is a compositional graphoid,  then \(P^*\) is Markovian to \(\phi(\mathcal{G};C)\).
\end{theorem}


Theorem \ref{correctnewswag} is shown by observing that the combined structural equilibrium model \eqref{swigSCM} is another structural equilibrium model, with the corresponding graph \(\phi(\mathcal{G}; C)\). Theorem \ref{correctgenant} is  used to show Theorem \ref{correctnewswag}. 

Note that the nodes merged in Step 2 of Algorithm \ref{newswagalg} correspond to counterfactual variables that are equivalent to observational variables, that is, \(X^{\textnormal{do}(C)}_i=X_i\); thus Theorem \ref{correctnewswag} describes the full joint distribution of observational and intervened variables. This merging step also avoids certain deterministic relationships between nodes in the graph \(\phi(\mathcal{G}; C)\), which can lead to violations of the faithfulness condition, leading to erroneous results as highlighted in \citet[Figure 15]{swig}. However, errors are still shared between observational and intervened varibles, as indicated by the bidirected edge \(3\leftrightarrow 3^{\textnormal{do}(C)}\) in \(\phi(\mathcal{G}; C)\) of Figure \ref{newswagfig}, which can potentially lead to violations of faithfulness. This motivates an alternative graphical representation inspired by single-world intervention graphs (SWIGs) from \cite{swig}. 

\subsubsection{Single-world interpretation}

Using Algorithm \ref{newswagalg}, the constructed \(\phi(\mathcal{G}; C)\) describes the relationship of both observational and interventional versions of the same variable entailed by the combined structural equilibrium model \eqref{swigSCM}, represented by having both \(i^{\textnormal{do}(C)}\) and \(i\) as nodes in the same graph.

However, outside of settings where one has a full description of the underlying data generating process (in our case, the structural equilibrium model), causal inference has the fundamental problem where only \emph{one} version of the same variable is observed \citep{fundamentalprob}. To accommodate this issue, we will also graphically represent the single-world counterfactual distribution consisting of only one version (observational or intervened) of each variable that is not in the treatment set \(C\), in line with \cite{swig}.

Let the \emph{posterior} of \(C\) in the observational graph \(\mathcal{G}\) be given by
     $$\textnormal{po}(C) =\{j\not\in C :  \textnormal{ there is a semi-directed path from } i\in C \textnormal{ to }  j \textnormal{ in } \mathcal{G}\}.$$ 

By marginalising the counterfactual graph \(\phi(\mathcal{G}; C)\) using \(\alpha_{\textnormal{m}}\) with \(M=\textnormal{po}(C)\), we obtain \(\alpha_{\textnormal{m}}(\phi(\mathcal{G}; C); \textnormal{po}(C))\), which we call \(\phi'(\mathcal{G}; C)\); see Figure \ref{subgex} for an example.

\begin{figure}[h]
\begin{center}
\centerline{\includegraphics[width=0.8\columnwidth]{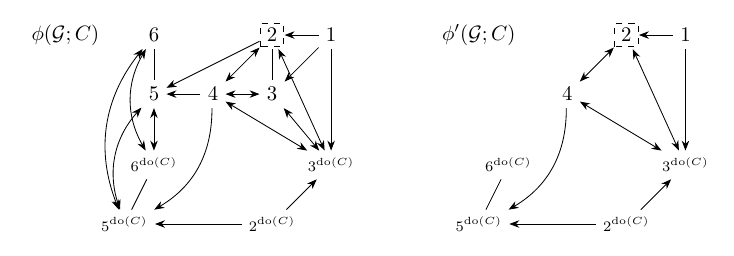}}
\caption{Left: \(C=2\) and \(\textnormal{po}(C) = \{3, 5, 6\}\).  Right: The single-world anterial interventional graph (SWAIG) \(\phi'(\mathcal{G}; C)\) obtained by marginalising \(\textnormal{po}(C)\) from \(\phi(\mathcal{G}; C)\).}
\label{subgex}
\end{center}
\end{figure}


Based on Theorem \ref{correctnewswag}, we have the following.

\begin{proposition}[Markov property of the single-world interpretation]\label{propswigcg}
Let the combined structural equilibrium model in \eqref{swigSCM} induce \(P^*\), the joint distribution over \(X_V\) and \(X_V^{\textnormal{do}(C)}\), and let \(P^*\) be a compositional graphoid. Let \(P^*_{\textnormal{sub}}\) be the distribution of variables in \(\phi'(\mathcal{G};C)\). The distribution \(P^*_{\textnormal{sub}}\) is Markovian to \(\phi'(\mathcal{G}; C)\).
 \end{proposition}
 


When \(\mathcal{G}\) is a DAG, marginalising \(\textnormal{po}(C)\) from \(\phi(\mathcal{G}; C)\) is equivalent to taking the subgraph over the remaining nodes not in \(\textnormal{po}(C)\), since there are no nodes in \(j\in \textnormal{po}(C)\) and \(k\not\in \textnormal{po}(C)\) such that \(j\rightarrow k\). Observe that SWIGs \citep{swig} relabel the descendants \(i\) of \(C\) as \(i^{\textnormal{do}(C)}\), and nodes \(j\in C\) get split into \(j^{\textnormal{do}(C\backslash j)}\) and \(j^{\textnormal{do}(C)}\). Thus, \(\phi'(\mathcal{G}; C)\) contains all variables in the SWIG \emph{except} the intervened variables of the form \(X_i^{\textnormal{do}(B)}\) where \(B\subseteq C\). 
%
%
%

In the case when the treatment set \(C=i\) is a singleton, it can be seen that the resulting graph \(\phi'(\mathcal{G}; C)\) contains \emph{all} the variables in the SWIG. Thus, we will refer to \(\phi'(\mathcal{G}; C)\), in the case when \(\mathcal{G}\)  is a general anterial graph, as a \emph{single-world anterial interventional graph} (SWAIG). Proposition \ref{propswigcg} then implies the Markov property of SWIGs with the additional condition that the joint distribution over observational and intervened variables is a compositional graphoid.

In Figure \ref{subgex}, note that there are no more bidirected edges between observational nodes \(i\) and \(j^{\textnormal{do}(C)}\) in \(\phi'(\mathcal{G}; C)\). This is because we consider only one version (observational or intervened) of each variable, hence common shared errors, and thus faithfulness violations, are avoided. However, note that, by including fewer variables, \(\phi'(\mathcal{G}; C)\) implies less conditional independencies than \(\phi(\mathcal{G}; C)\).

\begin{remark}[Taking into account of all nodes in the SWIG]\label{includingswig}
    When \(\mathcal{G}\) is a DAG, nodes of the form \(j^{\textnormal{do}(C\backslash j)}\) can be included by extending the construction of Algorithm \ref{newswagalg} based on parallel-world graphs \citep{multinetwork}; see section \ref{subappsec} to see how SWIGs are subgraphs of counterfactual graphs. \erk
\end{remark}

\section{Confounder Selection with Constraints}\label{tnm3}

The expression of common causal estimands such as the CATE and ETT in terms of the observational distribution depends on the \emph{consistency} condition \citep[Page 4]{exch} and the unconfoundedness assumption \eqref{defconfound}, with the consistency condition being satisfied when intervening structural equilibrium models. Since structural equilibrium models induce a joint distribution over observational and intervened variables via common errors based on \eqref{swigSCM}, the unconfoundedness assumption, and thus the notion of adjustment sets are well-defined. 

Thus, we will consider the problem of selecting a minimal adjustment set \(S\) for the confounding of treatment variable \(X_C\) on outcome variable \(X_O\)  subject to the following constraints $L$ and $U$:
\begin{align}\label{ccscons}
    L\subseteq S\subseteq U,
\end{align}
where \(L\) and \(U\) are fixed set constraints. 

These constraints may arise from some covariates that have to be excluded from \(S\) due to the economic and ethical costs of data collection, or covariates \(L\) that should be included in \(S\) such as prognostic variables following regulatory recommendations in adjusting for covariates from agencies like the FDA or EMA \citep[Section III]{fdaadjust}.


We assume we are given an anterial graph \(\mathcal{G}^*\) over the  nodes \(V^*\) such that the joint distribution of \(X^{\text{do(C)}}_O, X_C, X_U\), (and potentially other nuisance variables) is Markovian to \(\mathcal{G}^*\). We also assume that the treatment set \(C\) and outcome set \(O\) are singletons, which we emphasise using the notation \(c\) and \(o\). In the case of a structural equilibrium model corresponding to a chain-connected anterial graph \(\mathcal{G}\) that induces the distribution \(P^*\) over \(X_V\) and \(X_V^{\textnormal{do}(C)}\), Theorem \ref{correctnewswag} and Proposition \ref{propswigcg} imply that the counterfactual graph \(\phi(\mathcal{G};c)\) and the SWAIG \(\phi'(\mathcal{G}; c)\) are examples of the anterial graph \(\mathcal{G}^*\).

In practice, the graph \(\phi(\mathcal{G}; c)\), and thus \(\phi'(\mathcal{G}; c)\) after applying \(\alpha_\textnormal{m}\) to \(\phi(\mathcal{G}; c)\), can be obtained by applying Algorithm \ref{newswagalg} to the output of causal discovery algorithms that return an anterial graph \(\mathcal{G}\) from an input observational distribution \(P\)---these algorithms have been introduced by \cite{sad-meek}.

The assumption \eqref{defconfound} that the set \(S\) has to satisfy can then be expressed as a graphical separation of the graph \(\mathcal{G}^*\). In the case when \(\mathcal{G}^*\) is a maximal ancestral graph, standard approaches such as from \cite{construct} can then be applied to obtain \(S\).

 When \(\mathcal{G}^*\) is a general anterial graph, we present the element-wise procedure in Algorithm \ref{ccsalg}, taking the graph \(\mathcal{G}^*\) and the nodes \(o^{\textnormal{do}(c)}\) and \(c\) representing  \(X^{\text{do(c)}}_o\) and \(X_c\), respectively, as inputs to return such a set \(S\).

\begin{algorithm}[h]
   \caption{Constrained Confounder Selection}
   \label{ccsalg}
   {\bfseries Input:} anterial graph $\mathcal{G}^*$, nodes \(o^{\textnormal{do}(c)}, c\in V^*\), node subsets \(L, U\subseteq V^*\) (not containing \(c\) and \(o^{\textnormal{do}(c)}\))
   
 {\bfseries Output:} set of nodes \(S\)
\begin{algorithmic}[1]
   \State Obtain \(\alpha_{\textnormal{m}}(\alpha_{\textnormal{c}}(\mathcal{G}^*;L);V^*\backslash (U\cup \{c, o^{\textnormal{do}(c)}\}))\), call this \(\mathcal{G}^*_0\).
   \State Maximise \(\mathcal{G}^*_0\), by obtaining \(\text{max}(\mathcal{G}^*_0)\).
   \If{\(c\) is adjacent to \(o^{\text{do}(c)}\) in \(\text{max}(\mathcal{G}^*_0)\)}
   \State {\bfseries 
 return} no feasible set.
   \Else
   \State Let \(n=0\) and \(S=U\).
    \State \textbf{repeat until no such node}
   \State Select a node \(i\in S\backslash L\) such that \(c\) and \(o^{\textnormal{do}(c)}\) are not adjacent in \(\text{max}(\mathcal{G}^*_{n+1})\), where \(\mathcal{G}^*_{n+1}=\alpha_{\textnormal{m}}(\mathcal{G}^*_n; i)\).
   \State \(S\leftarrow S\backslash i\).
   
   \State \(n\leftarrow n+1\).
   \State \textbf{return} \(S\).
   \EndIf
\end{algorithmic}
\end{algorithm}
    \begin{figure*}[ht]
\begin{center}
\centerline{\includegraphics[]{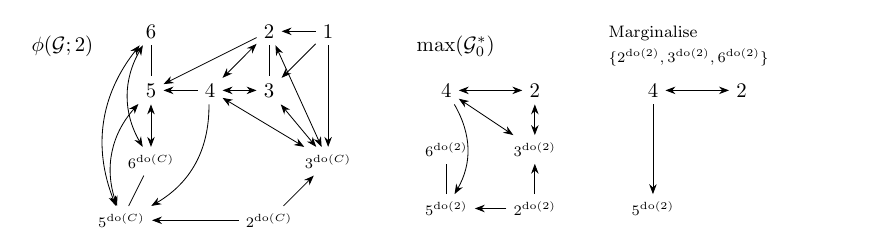}}
 \caption{Left: Graph \(\phi(\mathcal{G};2)\) from Figure \ref{newswagfig} is used as the input \(\mathcal{G}^*\) for Algorithm \ref{ccsalg}, with \(L=1\) and \(U=V^*\backslash \{2, 3,5,6, 5^{\textnormal{do}(2)}\}\), \(c\) and \(o^{\textnormal{do}(c)}\) being nodes \(2\) and \(5^{\textnormal{do}(2)}\) respectively. Middle: Graph \(\textnormal{max}(\mathcal{G}^*_0)\) obtained after conditioning on \(L\) marginalising \(V^*\backslash (U\cup \{c, o^{\textnormal{do}(c)}\})\) and maximising. The feasibility check passes since \(2\) is not adjacent \(5^{\text{do}(2)}\) in \(\textnormal{max}(\mathcal{G}^*_0)\). Right: Graph after Algorithm \ref{ccsalg} selects \(\{2^{\text{do}(2)},3^{\text{do}(2)}, 6^{\text{do}(2)}\}\) to marginalise and maximise. Since further marginalising \(4\) would introduce edge \(2\leftrightarrow 5^{\text{do}(2)}\), the algorithm terminates and returns \(\{1,4\}\) as \(S\).}
\label{algexfig}
\end{center}
\end{figure*}

Figure \ref{algexfig} illustrates how Algorithm \ref{ccsalg} works. After marginalising the nodes not in \(U\cup \{c, o^{\textnormal{do}(c)}\}\) and conditioning on the nodes in \(L\) (using operations \(\alpha_{\textnormal{m}}\) and \(\alpha_{\textnormal{c}}\) from Section \ref{bg}) to obtain \(\mathcal{G}^*_{0}\), Algorithm \ref{ccsalg} checks \(\textnormal{max}(\mathcal{G}^*_0)\) to determine if there is an adjustment set \(S\) that satisfies \eqref{ccscons}. Algorithm \ref{ccsalg} then removes nodes \(i\) from the upper set constraint \(U\) in a way that ensures, after removal, the existence of an adjustment set \(S\) where \(L \subseteq S\subseteq U\backslash i\). This removal is iteratively applied to the upper bound until no nodes can be removed and the remaining upper bound is returned as \(S\).


The general set-up of Algorithm \ref{ccsalg}, depending on the constraints \(L\) and \(U\), allows for the inclusion of intervened variables in the adjustment set \(S\). When \(S\) contains intervened variables, the obtained identification formula for causal estimands then contains terms from the interventional distribution and is not identifiable observationally. By excluding all the intervened variables from the set constraint \(U\) in Algorithm \ref{ccsalg}, only observational variables are allowed in the adjustment set \(S\). We will not enforce this constraint, since there are cases such as sequential exchangeability in \citet[Section 19.5]{exch} where \(S\) contains intervened variables.

Note that after obtaining \(\textnormal{max}(\mathcal{G}^*_0)\), if \(c\) is not adjacent to \(o^{\textnormal{do}(c)}\), Algorithm \ref{ccsalg} can in principle return the set \(L \cup \text{ant}(\{c, o^{\textnormal{do}(c)}\})\) in \(\textnormal{max}(\mathcal{G}^*_0)\) as \(S\) and terminate. Since nodes in \(\textnormal{max}(\mathcal{G}^*_0)\) (that are not \(c\) and \(o^{\textnormal{do}(c)}\)) are contained in \(U\), the returned \(S\) would also satisfy \(S\subseteq U\). However, in general, this does not lead to a \emph{minimal} adjustment set. Consider \(\textnormal{max}(\mathcal{G}^*_0)\) in Figure \ref{algexfig}, since \(L \cup 4\) is an adjustment set such that \(L \cup 4\subset L\cup \text{ant}(\{2, 5^{\textnormal{do}(2)}\})\), it can be seen that \(L\cup \text{ant}(\{2, 5^{\textnormal{do}(2)}\})\) is not minimal. Similarly, \(L\cup \text{pa}(\{c, o^{\textnormal{do}(c)}\})\) in \(\textnormal{max}(\mathcal{G}^*_0)\) is not minimal.

\begin{theorem}[Correctness of Algorithm \ref{ccsalg}]\label{CCSworks}Let \(\mathcal{G}^*\) and the single nodes \(c\) and \(o^{\textnormal{do}(c)}\) be the inputs of Algorithm \ref{ccsalg} and let the joint distribution of \(X^{\textnormal{do}(c)}_o, X_c,  \
X_U\), (and potentially other nuisance variables) be Markovian to the anterial graph \(\mathcal{G}^*\). When Algorithm \ref{ccsalg} returns a set \(S\), the set \(S\) 
    is an adjustment set 
    that satisfies \(L\subseteq S\subseteq U\).

    Furthermore, let the joint distribution of \(X_o^{\textnormal{do}(c)}, X_c, X_U, \) (and potentially other nuisance variables) be faithful to \(\mathcal{G}^*\). It then holds that \(S\) is minimal and if Algorithm \ref{ccsalg} 
    returns no feasible set, then there does not exist an adjustment set \(S\) 
    that satisfies \(L\subseteq S\subseteq U\).
\end{theorem}

Theorem \ref{CCSworks} states that if we have access to an anterial graph \(\mathcal{G}^*\) (e.g.\ \(\phi(\mathcal{G}; c)\)) to which the distribution over a set of variables that contain the observational and intervened variable of interest is Markovian to, when Algorithm \ref{ccsalg} returns a set \(S\), 
this set is an adjustment set \(S\) that satisfies the constraint \(L\subseteq S \subseteq U\). 
Furthermore, if this distribution is faithful to \(\mathcal{G}^*\),  then we can determine the existence of such a set \(S\) via the feasibility check of Algorithm \ref{ccsalg}, and that the returned adjustment set \(S\) is minimal.

\begin{remark}[Faithfulness and confounder selection]
    Note that the unconfoundedness assumption \eqref{defconfound}, and thus the notion of an adjustment set, does not refer to a graph and is purely probabilistic. Algorithm \ref{ccsalg} determines the existence of a graphical separation set and finds a  minimal graphical separation set in \(\mathcal{G}^*\). The graphical separation sets correspond to the adjustment sets only when \(\mathcal{G}^*\) is a faithful graph. \erk
\end{remark}

Note that in general, the maximising operation, \(\textnormal{max}\), is required for Algorithm \ref{ccsalg}, even if we assume maximality of the input graph \(\mathcal{G}^*\). As shown in Figure \ref{exmax}, a DAG which is always maximal can result in a non-maximal ancestral graph when marginalised (using \(\alpha_{\textnormal{m}}\)) causing the feasibility check using \(\mathcal{G}^*_0\) to be inaccurate.

\begin{figure}[h]
\begin{center}
\centerline{\includegraphics[]{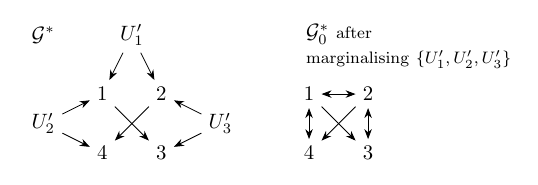}}
\caption{Left: Input anterial graph \(\mathcal{G}^*\) to algorithm \ref{ccsalg}, with \(L=\emptyset\), \(U=\{1,2,3,4\}\), and nodes \(C\) and \(O^{\textnormal{do}(C)}\) being \(3\) and \(4\). Note that \(\mathcal{G}^*\) is a DAG and is thus maximal. Right: Graph \(\mathcal{G}^*_0\) after marginalising over \(\{U'_1, U'_2, U'_3\}\) in Algorithm \(\ref{ccsalg}\) without maximising. We see that despite the absence of edges between nodes \(3\) and \(4\), there doesn't exist a set \(S\subseteq \{1,2,3,4\}\) such that \(3\perp_{\mathcal{G}^*} 4\cd S\).}
\label{exmax}
\end{center}
\end{figure}


To select an adjustment set \(S\) such that \(L\subseteq S\subseteq U\) for arbitrary set constraints \(L\) and \(U\), conventional approaches typically enumerate the adjustment sets efficiently \citep{construct} and implement the constraints post-hoc. By representing \(X_o^{\textnormal{do}(c)}\) and \(X_c\) in a Markovian graph \(\mathcal{G}^*\) (such as \(\phi(\mathcal{G};c)\) or \(\phi'(\mathcal{G}; c)\), obtained from the observational graph \(\mathcal{G}\)), when \(\mathcal{G}^*\) is a maximal ancestral graph, efficient approaches for identifying sets satisfying graphical separations such as from \cite{construct} can be used to select \(S\). When \(\mathcal{G}^*\) is a general anterial graph, Algorithm \ref{ccsalg} avoids constructing unfeasible adjustment sets by incorporating the constraints directly into the algorithm, and proceeds in an element-wise manner to select nodes to construct the minimal adjustment set.

An implementation of Algorithm \ref{ccsalg}, using the R package \verb|ggm| \citep{k-code} (developed for ancestral graphs) is hosted online: \cite{code}.


\begin{remark}[When \(C\) and \(O^{\textnormal{do}(C)}\) are node subsets]\label{nodesnotset} 
Note that the formulation of Algorithm \ref{ccsalg} requires the input \(C\) and \(O^{\textnormal{do}(C)}\) to be nodes. Consider the SWAIG \(\phi'(\mathcal{G}; C)\) in Figure \ref{fig:exset}; here $C$ is a two-element set.  It is not difficult to see that there does not exist a node subset \(S\) in \(\phi'(\mathcal{G}; C)\) such that \(C\perp_{\phi'(\mathcal{G}; C)} O^{\textnormal{do}(C)}\cd S\). Thus, in the case when the distribution over the nodes in \(\phi'(\mathcal{G}; C)\) is faithful to \(\phi'(\mathcal{G}; C)\), the absence of edges between \(C\) and \(O^\textnormal{do}(C)\) is not sufficient for the existence of an adjustment set, which is required for Algorithm \ref{ccsalg}. 

However, note that the converse is true, that is, presence of edges between subsets \(C\) and \(O^{\textnormal{do}(C)}\) implies the absence of an adjustment set. Thus, Algorithm \ref{ccsalg} is still sound, but not complete, in detecting when there are no feasible adjustment sets. \erk

\begin{figure}
    \centering
    \includegraphics[width=0.5\linewidth]{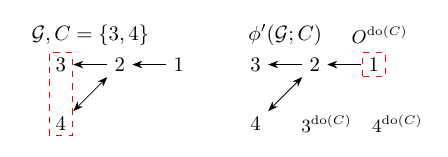}
    \caption{Left: Anterial graph \(\mathcal{G}\). Right: \(\phi'(\mathcal{G}; C)\) with \(O^{\textnormal{do}(C)}=1^{\textnormal{do}(C)}=1\) and \(C=\{3,4\}\).}
    \label{fig:exset}
\end{figure}
    
\end{remark}

\section{Conclusion and Future Work} 
We provide graphical models for structural equilibrium models (Definition \ref{antscm}), a generalisation of structural causal models, which also accounts for both equilibrium relationships and confounding. Given a structural equilibrium model, our graphical model is obtained using Algorithm \ref{obsalg}, and is based on anterial graphs. This generalises existing causal graphical models based on chain graphs \citep{chaingraphcausal}, ancestral graphs \citep{ancestral, Sad}, and thus DAGs.  



Following \cite{chaingraphcausal}, we use Gibbs samplers as the underlying dynamical processes to describe interventions on a treatment set \(C\) for structural equilibrium models. Using our anterial graphical model \(\mathcal{G}\) of the structural equilibrium model, these interventions can be represented graphically as \(\textnormal{do}_C(\mathcal{G})\) (using Algorithm \ref{intalg}). Our anterial graphical model \(\mathcal{G}\) also allows for the joint representation of observational and intervened variables in a single graph, as either \(\phi(\mathcal{G}; C)\) (via Algorithm \ref{newswagalg}) or \(\phi'(\mathcal{G}; C)\) (by marginalising \(\phi(\mathcal{G}; C)\)). The counterfactual graph \(\phi(\mathcal{G}; C)\) implies more counterfactual conditional independencies at the cost of potentially violating faithfulness, leading to errors similar to Figure 15 of \cite{swig}; the single-world interventional anterial graph \(\phi'(\mathcal{G}; C)\) avoids deterministic relationships \emph{and} common shared errors that lead to faithfulness violations, at the cost of encoding fewer conditional independencies.  

Our anterial graphical model \(\mathcal{G}\) is causally interpretable since interventions on the graphical model can be performed by only modifying and deleting edges directly adjacent to nodes in the treatment set \(C\). Due to this interpretability, algorithms involving interventions (Algorithms \ref{intalg} and \ref{newswagalg}) depend only on the corresponding observational graph \(\mathcal{G}\) and do not depend on the exact coupling of the error variables in the structural equilibrium model. 

We also provided an element-wise procedure to select an adjustment set \(S\) subject to \(L\subseteq S\subseteq U\) with prescribed set constraints \(L\) and \(U\) (Algorithm \ref{ccsalg}). This provably returns the minimal adjustment set under conditions such as when the representation \(\phi'(\mathcal{G}; C)\) introduced is faithful. Unlike conventional approaches, our method is valid for anterial graphs and avoids constructing unfeasible sets during the algorithm.


Theorems \ref{correctgenant}, \ref{intmarkov}, and \ref{correctnewswag} depend on the corresponding distributions being compositional graphoids, which is an assumption that is automatically made when one considers faithfulness \citep{kayvanfaith}. Analogous to DAGs, chain graphs and ancestral graphs, the compositional graphoid condition can potentially be relaxed by defining a new notion of a local Markov property for the class of anterial graphs, which implies the Markov property without any strong further conditions.

We also note that in Algorithm \ref{ccsalg}, after checking for feasibility, only the marginalisation operation \(\alpha_{\textnormal{m}}\) is used iteratively. However, allowing the conditioning operation \(\alpha_{\textnormal{c}}\) to be used alongside \(\alpha_{\textnormal{m}}\) may improve the algorithm. Consider Figure \ref{algexfig}, the minimal adjustment set \(L\cup 4\) can be selected in just 2 steps by applying the marginalising operation \(\alpha_{\textnormal{m}}\) and the conditioning operation \(\alpha_{\textnormal{c}}\) to the middle graph of Figure \ref{algexfig} and maximising to obtain the graph in Figure \ref{fig:altlgo}. However, in general, conditioning may introduce unnecessary nodes into \(S\), resulting in a \(S\) that is no longer minimal. Thus, speeding up Algorithm \ref{ccsalg} by allowing for both \(\alpha_{\textnormal{m}}\) and \(\alpha_{\textnormal{c}}\) operations, while ensuring the minimality of the selected \(S\) is a subject of future interest.

In Section \ref{sec:counterfactual}, we focused on counterfactual conditional independencies that relate observational and interventional variables where the intervened value is fixed, such as the unconfoundedness assumption \eqref{defconfound}. Extending this framework to capture more general counterfactual assumptions involving multiple intervention values is an important direction for future work. A potential approach is to use the ``templates,"  proposed by \cite{swig}, to systematically handle different intervened values.

\begin{figure}[h]
    \centering
    \centerline{\includegraphics[scale=0.9]{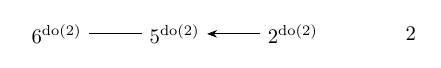}}
    \caption{Graph obtained after marginalising \(\{3^{\text{do}(2)}\}\), conditioning on \(4\) and maximising the middle graph \(\textnormal{max}(\mathcal{G}^*_0)\) from Figure \ref{algexfig}. Since there is no path from \(2\) to \(5^{\text{do}(2)}\), it is clear that  \(2\perp 5^{\text{do}(2)}\cd  L\cup 4=\{1,4\}\), and \(L \cup 4\) can be returned as a minimal adjustment set. }
    \label{fig:altlgo}
\end{figure}

As can be seen in Remark \ref{nodesnotset}, Algorithm \ref{ccsalg} is only sound, but not complete, in detecting the absence of feasible adjustment sets when the inputs \(C\) and \(O^{\textnormal{do}(C)}\) are subset of nodes. Thus, extending Algorithm \ref{ccsalg} to the multivariate setting where \(C\) and \(O^{\textnormal{do}(C)}\) are subsets of nodes remains to be investigated.

\section{Acknowledgments}
We thank Christopher Meek and Thomas Richardson for their insightful remarks and advice.


\begin{appendices}
\section{Appendix and Proofs}
\subsection{Proof of Theorem \ref{correctgenant}}

We will use the following from \cite{lkayvan} and \cite{Sad} to show Theorem \ref{correctgenant}. 

\begin{definition}[Chain mixed graph \citep{lkayvan}]\label{def cmg}
    A chain mixed graph \(\mathcal{G}\) over a set of  nodes \(V\), is a graph which may contain directed,  undirected, 
    and bidirected edges, such that there does not exist semi-directed cycles in \(\mathcal{G}\).
\end{definition}
Anterial graphs are thus a subclass of chain mixed graphs where Item 1 of Definition \ref{def: ant} does not have to hold. 
Analogous to inducing paths in an ancestral graph, the notion of a \emph{primitive inducing path} will be used to characterise the  maximality of chain mixed graphs. 
\begin{definition}[Primitive inducing path]
     A primitive inducing path connecting  the nodes \(i\) and \(j\) is 
     a path \(\pi=\langle i, q_1,\ldots
    , q_n, j \rangle\) such that:
    \begin{itemize}
        \item \(q_1,\ldots, q_n\in \text{ant}(i,j)\), 
        \item the edges between subsequent nodes \(q_1,
        \ldots, q_n\) in \(\pi\) are either bidirected or undirected, and
        \item the edge between \(i\) and \(q_1\) is either bidirected or \(i\rightarrow q_1\), and the edge between \(j\) and \(q_n\) is either bidirected or \(j\rightarrow q_n\).
    \end{itemize}
\end{definition}

We will use the notation \(i\not\sim_p j\) to denote that the nodes \(i\) and \(j\) are not connected by a primitive inducing path.
%
We will use the following version of \cite{lkayvan} for our proof.

\begin{proposition}[\citep{lkayvan}] \label{intproof2}
    A chain mixed graph \(\mathcal{G}\) is maximised by adding edges between non-adjacent nodes that are connected by a primitive inducing path. Call the resulting maximal graph \(\textnormal{max}(\mathcal{G})\).
    Furthermore, the anterior of \(\mathcal{G}\)  is preserved; that is, \textnormal{ant}\((i)\) in \(\mathcal{G}\) is the same as \textnormal{ant}\((i)\) in \(\textnormal{max}(\mathcal{G})\), for every node \(i\).
\end{proposition}


\begin{definition}[Pairwise Markov property \citep{lkayvan}]
    The distribution \(P\) satisfies the pairwise Markov property w.r.t.\ the  chain mixed graph \(\mathcal{G}\) if for nodes \(i\) and \(j\), we have
    \begin{align}\label{pairdef}
        i \text{ not adjacent to } j \text{ in }\mathcal{G}\Rightarrow i\ci j\cd \text{ant}(i,j).
    \end{align} 
\end{definition}

\begin{proposition}[Theorem 4 \citep{lkayvan}]\label{pairtoglobal}
    Let \(\mathcal{G}\) be a maximal chain mixed graph, and the distribution \(P\) be a compositional graphoid. The following equivalence holds.
    \begin{align*}
        P \text{ satisfies the pairwise Markov property w.r.t. }\mathcal{G} \iff P \text{ is Markovian to } \mathcal{G}.
    \end{align*}
\end{proposition}

When \(\mathcal{G}\) is not maximal, we have the following version of Proposition \ref{pairtoglobal}.

\begin{proposition}[Proposition \ref{pairtoglobal} for non-maximal \(\mathcal{G}\)]\label{nonmaxpairtoglob}
    Let \(\mathcal{G}\) be a chain mixed graph, and the distribution \(P\) be a compositional graphoid. The following are equivalent.
    \begin{enumerate}
        \item \(P\) is Markovian to \(\mathcal{G}\).
        \item For non-adjacent nodes \(i\) and \(j\) in \(\mathcal{G}\) such that \(i\not\sim_p j\), it holds that \(i\ci j\cd \textnormal{ant}(i,j)\).
    \end{enumerate}
\end{proposition}
\begin{proof} 
By Proposition \ref{intproof2}, we have
     \begin{align}\label{working}
i \textnormal{ not adjacent to } j \textnormal{ in max}(\mathcal{G})&\iff  i \textnormal{ not adjacent to } j \textnormal{ and }i \not\sim_p j \textnormal{ in } \mathcal{G}, \text{ and} \nonumber\\
i \ci j \cd \textnormal{ant}(i,j) \textnormal{ (anterior of max}(\mathcal{G})) &\iff i \ci j \cd \textnormal{ant}(i,j) \textnormal{ (anterior of }\mathcal{G}).
\end{align}
Thus, by replacing both sides of the implication in \eqref{pairdef} with the correspondence in \eqref{working}, \(P\) satisfying the pairwise Markov property w.r.t. max\((\mathcal{G})\) is equivalent to Item 2.

By Proposition \ref{pairtoglobal} and since \(\mathcal{G}\) and max\((\mathcal{G})\) are Markov equivalent, we have
\begin{align*}
    P \textnormal{ satisfies the pairwise Markov property w.r.t. max}(\mathcal{G}) &\iff P \textnormal{ is Markovian to max}(\mathcal{G})\\
    &\iff P \textnormal{ is Markovian to } \mathcal{G} \textnormal{ (Item 1)}. \qedhere
\end{align*}
\end{proof}

We will also use \citet[Theorem 27]{Sad}.

\begin{proposition}[\citep{Sad}]\label{Sadmarkov}
    Let \(\mathcal{G}\) be an ancestral graph over the nodes \(V\) and \(P\) be a joint distribution over \(V\) induced from an SCM that corresponds to \(\mathcal{G}\). The distribution \(P\) is Markovian to \(\mathcal{G}\).
\end{proposition}

\begin{remark}\label{sadmarkovremark}
    The proof of Proposition \ref{Sadmarkov} uses Markov property results for acyclic directed mixed graphs from \cite{richadmg} which do not depend on the maximality of graph \(\mathcal{G}\). \erk
\end{remark}

We first show that \(J^\tau(\cdot\cd x_{\textnormal{pa}(\tau)})\), the law of \( f_{\tau}(x_{\text{pa}(\tau)}, \epsilon_{\tau})\) of the structural equilibrium model, satisfies Lemma \ref{localind}.

\begin{lemma}\label{localind}
    Let \(\mathcal{G}\) be the corresponding graph (constructed using Algorithm \ref{obsalg}) of a structural equilibrium model. Consider the chain component \(\tau\) and associated distribution \(J^\tau(\cdot\cd x_{\textnormal{pa}(\tau)})\). For every node \(i\), the conditional marginal \(J^{\tau}_{i}(\cdot \cd x_{\tau\backslash i}; x_{\textnormal{pa}(\tau)})\) does not depend on \(x_j\)  for  \(j\in \big(\tau \cup \textnormal{pa}(\tau)\big)\backslash \big(\textnormal{pa}(i)\cup \textnormal{ne}(i)\cup i\big)\).
\end{lemma}
\begin{proof}
Consider 
a 
node \(j\in \tau\backslash (\textnormal{ne}(i)\cup i)\). By construction, in Step 1 of Algorithm \ref{obsalg}, \(J^{\tau}_i(\cdot \cd x_{\textnormal{pa}(\tau)})\) satisfies  \(i\ci j\cd \tau\backslash\{i,j\}\), which implies that \(J^{\tau}_{i}(\cdot \cd x_{\tau\backslash i}; x_{\textnormal{pa}(\tau)})\) does not depend on \(x_j\).

Consider 
a
node \(j\in \textnormal{pa}(\tau)\backslash \textnormal{pa}(i)\). 
By Step 2 of  Algorithm \ref{obsalg}, the distribution $J^{\tau}_{i}(\cdot \cd x_{\tau\backslash i}; x_{\textnormal{pa}(\tau)})$ does not depend on $x_j$.
%
%
\end{proof}

For each node $i$, the conditional marginal \(J^\tau_i(\cdot\cd x_{\tau\backslash i};x_{\textnormal{pa}(\tau)})\)  depends only on the adjacent nodes in  \(\textnormal{pa}(i)\cup \textnormal{ne}(i)\).   

Let sant\((i)\) = ant\((i)\backslash \tau(i)\) be the strict anterior of the  node \(i\) and
let sant\((i,j)\)=  \(\textnormal{ant}(i,j)\backslash \big(\tau(i)\cup \tau(j)\big)\).

\begin{proof}[Proof of Theorem \ref{correctgenant}]
Let \(\mathcal{G}\) be a chain-connected anterial graph, and \(i\) and \(j\) be non-adjacent nodes such that \(i\not\sim_p j\). By Proposition \ref{nonmaxpairtoglob}, it suffices to show that the distribution \(P\) induced from the structural equilibrium model satisfies \(i\ci j\cd \text{ant}(i,j)\). 

We will consider the following cases.

\begin{enumerate}
    \item 
    \label{prev-arg}
    Suppose that \(j\in \textnormal{ant}(i)\). We first consider \(P_{\tau(i)}(\cdot \cd x_{\textnormal{sant}(i)})\). Since \(\text{ant}(i)\backslash \textnormal{sant}(i)= \tau(i)\backslash i\), we then further condition \(P_{\tau(i)}(\cdot \cd x_{\textnormal{sant}(i)})\) given \(x_{\tau(i)\backslash i}\) to obtain \(P_i(\cdot \cd x_{\text{ant}(i,j)\cup j})=P_i(\cdot \cd x_{\text{ant}(i)})\), since \(j\in \textnormal{ant}(i)\), and show that \(P_i(\cdot \cd x_{\text{ant}(i)})\) does not depend on \(x_j\).
    
    Since there do not exist bidirected edges between nodes in sant\((i)\) and nodes in \(\tau(i)\) by the anterial assumption, we have \(\epsilon_{\tau(i)}\ci X_{\text{sant}(i)}\).
    By definition, we have \(X_{\tau(i)}=f_{\tau(i)}(X_{\text{pa}(\tau(i))}, \epsilon_{\tau(i)})\), and since the distribution of \(\epsilon_{\tau(i)}\) remains the same after conditioning given \(x_{\text{sant}(i)}\),  we have 
    \begin{align}\label{argue}
        f_{\tau(i)}(x_{\text{pa}(\tau(i))}, \epsilon_{\tau(i)})&\sim P_{\tau(i)}(\cdot \cd x_{\textnormal{sant}(i)})= J^{\tau(i)}(\cdot\cd x_{\text{pa}(\tau(i))}).
    \end{align}
    Further conditioning on  \(x_{\tau(i)\backslash i}\), we obtain
      \begin{align}\label{proof1}
          P_i(\cdot \cd x_{\textnormal{ant}(i)})=J^{\tau(i)}_i(\cdot\cd x_{\text{pa}(\tau(i))\cup \tau(i) \large \backslash i}).
      \end{align}

      Since \(i\) and \(j\) are not adjacent, \(j\not \in \text{ne}(i)\cup \text{pa}(i)\). By \eqref{proof1} and Lemma  \ref{localind}, 
      \(P_i(\cdot \cd x_{\textnormal{ant}(i)})\) does not depend on \(x_j\).
      
    \item  Suppose that \(j\not \in \textnormal{ant}(i)\). If \(i\in \text{ant}(j)\), 
    then our previous argument \eqref{prev-arg} applies. Thus we also assume that \(i\not \in \text{ant}(j)\). 
    
    For the chain-connected anterial graph \(\mathcal{G}\), consider the collapsed graph \(\mathcal{G}_{\textnormal{col}}\) with nodes representing chain components;  if \(i\rightarrow j\) or \(i\leftrightarrow j \) in \(\mathcal{G}\), then we set \(\tau(i)\rightarrow \tau(j)\) or \(\tau(i)\leftrightarrow \tau(j) \) in \(\mathcal{G}_{\textnormal{col}}\),
    respectively. See Figure \ref{proofillust} for an example.
    
    \begin{figure}[h]
        \centering
        \includegraphics[width=0.5\linewidth]{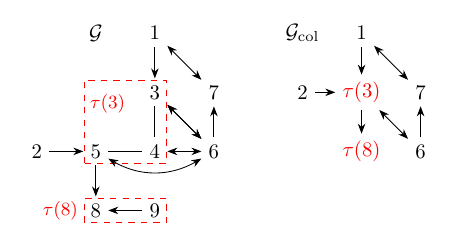}
        \caption{Chain-connected anterial graph \(\mathcal{G}\) and collapsed graph \(\mathcal{G}_{\textnormal{col}}\) representing chain components of \(\mathcal{G}\).}
        \label{proofillust}
    \end{figure}
    We make the following observations about \(\mathcal{G}_{\textnormal{col}}\).
    \begin{enumerate}
        \item
        The graph
        \(\mathcal{G}_{\textnormal{col}}\) does not have undirected edges and is an ancestral graph. This is because semi-directed paths from \(k\) and \(\ell\) and bidirected edges \(k\leftrightarrow \ell\) in \(\mathcal{G}\) correspond to directed paths from \(\tau(k)\) to \(\tau(\ell)\) and bidirected edges \(\tau(k)\leftrightarrow \tau(\ell)\) (since \(\mathcal{G}\) is chain-connected) in \(\mathcal{G}_\textnormal{col}\) respectively. 
        \item
        \label{transfer-prim}
        For primitive inducing paths, we also have
        \(i\not\sim_p j\) in \(\mathcal{G}\) implies that \(\tau(i)\not\sim_p \tau(j)\) in \(\mathcal{G}_{\textnormal{col}}\); to see this, we will show the contrapositive. 
        
        Let \(\langle \tau(i), \tau(q_1),\ldots
    , \tau(q_n), \tau(j) \rangle\) be a primitive inducing path in \(\mathcal{G}_{\textnormal{col}}\). 
    %
    We must have \(\tau(i)\leftrightarrow \tau(q_1)\) and \(\tau(j)\leftrightarrow \tau(q_n)\),  
    since if \(\tau(i)\rightarrow \tau(q_1)\) in \(\mathcal{G}_{\textnormal{col}}\), then in order for \(\mathcal{G}_{\textnormal{col}}\) to not have directed cycles, \(\tau(q_1)\) must be in \(\textnormal{ant}(\tau(j))\) in \(\mathcal{G}_{\textnormal{col}}\), which implies \(i\in \textnormal{ant}(j)\) in \(\mathcal{G}\), a contradiction. The argument for \(\tau(j)\leftrightarrow \tau(q_n)\) follows similarly.
    Since \(\mathcal{G}\) is chain-connected, 
    the path \(\langle i, q_1,\ldots
    , q_n, j \rangle\) connected by bidirected edges with \(q_1,\ldots, q_n\in \textnormal{ant}(i,j)\) in \(\mathcal{G}\), so that it is a primitive inducing path.
    \end{enumerate}

    Let \(\textnormal{pa}_{\textnormal{col}}(\tau)\) denote the parents of node \(\tau\) in \(\mathcal{G}_{\textnormal{col}}\). By definition of the structural equilibrium model of \(\mathcal{G}\), we can re-express the induced random variables \(X_V\) equivalently as 
    \begin{align*}
        X_{\tau}=f_{\tau}(X_{\textnormal{pa}_{\textnormal{col}}(\tau)}, \epsilon_\tau)
    \end{align*}
    for every node \(\tau\) in \(\mathcal{G}_{\textnormal{col}}\). We see that this is an SCM that corresponds to
    the
    ancestral graph \(\mathcal{G}_{\textnormal{col}}\). Applying Proposition \ref{Sadmarkov}, we have that the distribution of the induced random variables \(X_{\tau_1}, \ldots, X_{\tau_n}\) (where \(\tau_1, \ldots, \tau_n\) are seen as nodes in \(\mathcal{G}_{\textnormal{col}}\)) is Markovian to \(\mathcal{G}_{\textnormal{col}}\).   

    By \eqref{transfer-prim},  we have \(\tau(i)\not\sim_p \tau(j)\) in \(\mathcal{G}_{\textnormal{col}}\). Using the pairwise Markov property of \(\mathcal{G}_{\textnormal{col}}\), we obtain \(\tau(i)\ci \tau(j)\cd \textnormal{sant}(i,j)\), from which we have \(i\ci j\cd \textnormal{ant}(i,j)\), via weak union \citet[Section 4]{dawidrss-ci}. \qedhere
    \end{enumerate}
\end{proof}

\begin{remark}
    Note that proof techniques for the case of DAGs and chain graphs use local Markov properties \citep[Section 3.2.3]{lauritzenbook}, which to the best of our knowledge, are unavailable for anterial graphs. \erk 
\end{remark}

\subsection{Proof of Theorem  \ref{intmarkov}}

Consider an observational structural equilibrium model with the corresponding graph \(\mathcal{G}\). After intervening on a treatment set \(C\), we will show that after applying Algorithm \ref{obsalg} to the intervened structural equilibrium model \eqref{intscm}, the returned corresponding graph \(\mathcal{G_\textnormal{do}}\)  is, in fact, the graph \(\textnormal{do}_C(\mathcal{G})\) obtained from Algorithm \ref{intalg} (Proposition \ref{intrel}). See Figure \ref{goodantfig}.
\begin{figure}[h]
    \centering
    \begin{tikzpicture}[]

\node (e1) at (0,-1.5) {\(\mathcal{G}\)};
\node (e5) at (4,-1.5) {\(\textnormal{do}_C(\mathcal{G})\)};
\node (e2) at (0,0) {Observational model};
\node (e3) at (7,0) {Intervened model};
\node (e4) at (7,-1.5) {\(\mathcal{G}_\textnormal{do}\) };
\draw[->] (e1) -- (e5) node[midway, above] {Algorithm \ref{intalg}};

\draw[<-] (e1) -- (e2) node[midway, left] {Corresponding graph};
\draw[->] (e2) -- (e3) node[midway, above] { intervene };
\draw[->] (e3) -- (e4) node[midway, right] { Corresponding graph };
\draw[-, draw=none] (e5) -- (e4) node[midway] { equals to} node[midway, above=0.3cm] { \small Proposition \ref{intrel}};
\node (e1) at (0,0) {};
    \end{tikzpicture}
    \caption{Proof schematics of Theorem \ref{intmarkov}.}
    \label{goodantfig}
\end{figure}
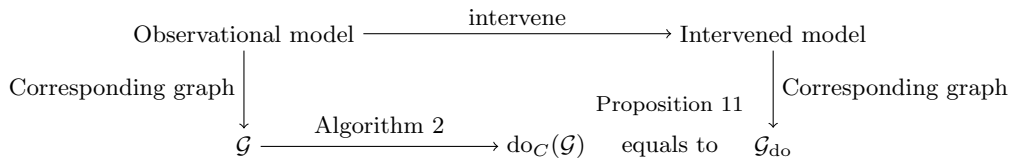

For a treatment set \(C\) and a chain component \(\tau\) of \(\mathcal{G}\) where \(\tau\cap C \neq \emptyset\), recall the functional assignment \(X_\tau=f^{\text{do}(C)}_\tau(X_{\text{pa}(\tau)}, \epsilon_\tau)\) of the intervened structural equilibrium model \eqref{intscm}. From \eqref{keyeqint}, we can see that \(J^{\tau, \textnormal{do}(C)}(\cdot\cd x_{\text{pa}(\tau)})\), the law of \(f^{\text{do}(C)}_\tau(x_{\text{pa}(\tau)}, \epsilon_\tau)\), is 
    \(J^\tau_{\tau\backslash C}(\cdot\cd a_{ \tau\cap C};x_{\text{pa}(\tau)})\delta_{ \tau\cap C}(a_{\tau\cap C})\).  Recall that $a_C$ are the intervened values.

We can decompose the functional assignment \(X_\tau=f^{\text{do}(C)}_\tau(X_{\text{pa}(\tau)}, \epsilon_\tau)\) in \eqref{intscm} separately into the parts \(\tau\cap C\) and \(\tau\backslash C\) as
\begin{align}\label{multicompreex}
    X_{\tau\cap C}&=a_{\tau\cap C} \text{ and } \nonumber\\
    X_{\tau\backslash C}&=f^{\textnormal{do}(C)}_{\tau\backslash C}(X_{\textnormal{pa}(\tau)}, X_{ \tau\cap C}, \epsilon_{\tau}),
\end{align}
where \(f^{\textnormal{do}(C)}_{\tau\backslash C}(x_{\textnormal{pa}(\tau)}, x_{\tau\cap C}, \epsilon_{\tau})\sim J^\tau_{\tau\backslash C}(\cdot\cd a_{ \tau\cap C};x_{\text{pa}(\tau)})\) with the same error \(\epsilon_\tau\). We will use the expression \eqref{multicompreex} to construct the corresponding graph \(\mathcal{G}_\textnormal{do}\) and show Proposition \ref{intrel}.

Given a node \(i\), let \(\textnormal{pa}_\textnormal{do}(i)\) and \(\textnormal{ne}_\textnormal{do}(i)\) denote the parents or neighbours of \(i\) in \(\mathcal{G}_\textnormal{do}\). 

\begin{proposition}
The corresponding graph \(\mathcal{G}_\textnormal{do}\) of the intervened structural equilibrium model \eqref{intscm} is \(\textnormal{do}_C(\mathcal{G})\).\label{intrel}
    
\end{proposition}
\begin{proof}
Since the graphs \(\mathcal{G}_\textnormal{do}\) and \(\textnormal{do}_C(\mathcal{G})\) have the same nodes, it suffices to show that both graphs have the same set of undirected, directed, and bidirected edges. 

To show that the sets of undirected and directed edges coincide, given  a node \(i\in \tau\) for some chain component \(\tau\) in \(\mathcal{G}\), we show how \(\textnormal{pa}_\textnormal{do}(i)\) and \(\textnormal{ne}_\textnormal{do}(i)\) is related to \(\textnormal{pa}(i)\) and \(\textnormal{ne}(i)\), the parents and neighbours of \(i\) in \(\mathcal{G}\), as follows.
\begin{itemize}
        \item If \(i\not\in C\), then \(\textnormal{ne}_{\textnormal{do}}(i)=\textnormal{ne}(i)\backslash C\). If \(i\in C\), then \(\textnormal{ne}_{\textnormal{do}}(i)=\emptyset\).

         For nodes \(i\in \tau\backslash C\), Step 1 of Algorithm \ref{obsalg} uses \(J^\tau_{\tau\backslash C}(\cdot\cd a_{ \tau\cap C};x_{\text{pa}(\tau)})\) to construct \(\textnormal{ne}_{\textnormal{do}}(i)\) as the set of nodes \(j\in \tau\backslash C\) satisfying
\begin{align*}
    i\ci j\cd (\tau\backslash C)\backslash\{i,j\}  \textnormal{ for the conditional distribution } J^\tau_{\tau\backslash C}(\cdot\cd a_{ \tau\cap C};x_{\text{pa}(\tau)}).
\end{align*}
This is equivalent to \(j\) satisfying \(i\ci j\cd \tau\backslash\{i,j\}\) for \(J^\tau(\cdot\cd x_{\text{pa}(\tau)})\), which is how Algorithm \ref{obsalg} determines whether \(j\in \tau\backslash C\) should be in \(\textnormal{ne}(i)\) when using \(J^\tau(\cdot\cd x_{\text{pa}(\tau)})\). Thus \(\textnormal{ne}_{\textnormal{do}}(i)=\textnormal{ne}(i)\cap (\tau\backslash C)=\textnormal{ne}(i)\backslash C\).

For \(i\in C\), applying Step 1 of  Algorithm \ref{obsalg} to the law \(\delta_i(a_i)\) of the functional assignment \(X_{\tau\cap C}=a_{\tau\cap C}\) results in \(\textnormal{ne}_\textnormal{do}(i)=\emptyset\). 
        \item If \(i\not\in C\), then \(\textnormal{pa}_{\textnormal{do}}(i)=\textnormal{pa}(i)\cup (\textnormal{ne}(i)\cap C)\). If \(i\in C\), then \(\textnormal{pa}_\textnormal{do}(i)=\emptyset\). 

        For nodes \(i\in \tau\backslash C\), Step 2 of Algorithm \ref{obsalg} uses \(J^\tau_{\tau\backslash C}(\cdot\cd a_{ \tau\cap C};x_{\text{pa}(\tau)})\) to 
        select
        \(\textnormal{pa}^\textnormal{do}(i)\) as the set of nodes \(j\in \textnormal{pa}(\tau)\cup (\tau \cap C)\) such that for every value \(x_{(\tau\cup \textnormal{pa}(\tau))\backslash \{i,j\}}\), we have
     \begin{align*}
         J^{\tau}_i(\cdot \cd x_{(\tau\backslash C)\backslash i}, a_{\tau\cap C}; x_{\textnormal{pa}(\tau)})=J^{\tau}_i(\cdot \cd x_{\tau\backslash i}; x_{\textnormal{pa}(\tau)})
     \end{align*}
     depends on the value of \(x_j\). This is how Algorithm \ref{obsalg} determines whether \(j\in \textnormal{pa}(\tau)\) should be in \(\textnormal{pa}(i)\). For nodes \(j\in \tau\cap C\), from Lemma \ref{localind}, it can be seen that \(J^{\tau}_i(\cdot \cd x_{\tau\backslash i}; x_{\textnormal{pa}(\tau)})\) depends only on nodes \(j\in  \textnormal{ne}(i)\). Thus \(\textnormal{pa}_{\textnormal{do}}(i)=\textnormal{pa}(i)\cup (\textnormal{ne}(i)\cap C)\). 

     For \(i\in C\), applying Step 2 of Algorithm \ref{obsalg} to the law \(\delta_i(a_i)\) of the functional assignment \(X_{\tau\cap C}=a_{\tau\cap C}\) results in \(\textnormal{pa}_\textnormal{do}(i)=\emptyset\). 
    \end{itemize}
     Observe that this relationship is also exactly described by Algorithm \ref{intalg}. 
     For example, 
     from Figure \ref{intalgex}, it can be seen that \(\textnormal{ne}_{\textnormal{do}}(5)=\{3\}\backslash 3=\textnormal{ne}(5)\backslash C\) and \(\textnormal{pa}_\textnormal{do}(5)=\emptyset \cup 2=\textnormal{pa}(5)\cup \big(\textnormal{ne}(5)\cap C\big)\).
      
     We next show that the bidirected edges coincide. For each part \(\tau\) of \(\mathcal{G}\), the errors \(\epsilon_\tau\) for the functional assignment of \(X_{\tau\backslash C}\) in \eqref{multicompreex} have the same joint distribution, and thus dependencies, as the errors in the observational structural equilibrium model. Thus applying Step 3 of Algorithm \ref{obsalg} to the errors \(\epsilon_\tau\) of \eqref{multicompreex} returns the same bidirected edges as \(\mathcal{G}\) only with the bidirected edges connected to the nodes in \(\tau\cap C\) removed (since the functional assignment of \(X_{\tau\cap C}\) in \eqref{multicompreex} is deterministic). These remaining edges coincide with the bidirected edges of \(\textnormal{do}_C(\mathcal{G})\).
%
%
\end{proof}



\begin{lemma}\label{do is chain}
    Let \(\mathcal{G}\) be a chain-connected anterial graph. Then \(\text{do}_C(\mathcal{G})\) is  a chain-connected anterial graph.
\end{lemma}
\begin{proof}
We make the following observations.
\begin{enumerate}
    \item 
    \label{Ob-1}
    Edges with endpoint nodes not in \(C\) are the same in both \(\mathcal{G}\) and \(\text{do}_C(\mathcal{G})\).
    \item
    \label{Ob-2}
    Semi-directed paths and cycles (if they exist) in \(\text{do}_C(\mathcal{G})\) cannot contain nodes in \(C\).
    \item 
    \label{Ob-3}
    Nodes \(i\) in \(C\) are only connected to directed edges pointing away from \(i\).
\end{enumerate}

Recall Definition \ref{def: ant}; we first verify that the graph is anterial.
 Observations \ref{Ob-1} and \ref{Ob-2} imply that \(\text{do}_C(\mathcal{G})\) does not contain any semi-directed cycle in \(\mathcal{G}\).
    Towards a contradiction, suppose that there exists a bidirected edge \(j\leftrightarrow k\) with semi-directed path between \(j\) and \(k\) in \(\text{do}_C(\mathcal{G})\), then Observation \ref{Ob-3} implies that \(j,k\not \in C\), from  which, using Observation \ref{Ob-1}, we have that bidirected edge \(j\leftrightarrow k\) is in \(\mathcal{G}\). Observations \ref{Ob-1} and \ref{Ob-2} imply the same semi-directed path between endpoints of \(j\leftrightarrow k\) in \(\mathcal{G}\), which is absurd, since \(\mathcal{G}\) is anterial.

Recall Definition \ref{defchain};  we finally verify that the graph is chain-connected.  Towards a contradiction,  suppose that there exists node configuration \(i\leftrightarrow j-k\) in \(\text{do}_C(\mathcal{G})\).  From observation 3, \(i,j\) and \(k\) cannot be in \(C\), thus observation 1 implies that the node configuration \(i\leftrightarrow j-k\) exists in \(\mathcal{G}\), which implies that \(i\leftrightarrow k\) in \(\mathcal{G}\) and thus by observation 1 in \(\textnormal{do}_C(\mathcal{G})\) as well.
\end{proof}

\begin{proof}[Proof of Theorem \ref{intmarkov}]

By Proposition \ref{intrel}, we have that \(\textnormal{do}_C(\mathcal{G})\) is the corresponding graph of the intervened structural equilibrium model \eqref{intscm}. By Lemma \ref{do is chain}, we have that \(\textnormal{do}_C(\mathcal{G})\) is a chain-connected anterial graph. Thus an application of  Theorem \ref{correctgenant} completes the proof. 
\end{proof}
\subsection{Proof of   Theorem \ref{correctnewswag}}
Consider an observational structural equilibrium model with the corresponding graph \(\mathcal{G}\). Similar to the proof of Theorem \ref{intmarkov}, we first show that after applying Algorithm \ref{obsalg} to the combined structural equilibrium model \eqref{swigSCM}, the returned corresponding graph \(\mathcal{G}_\phi\) is, in fact, the graph \(\phi(\mathcal{G}; C)\) obtained from Algorithm \ref{intalg}. 
%
%

\begin{proposition}\label{swigrel}
    The corresponding graph \(\mathcal{G}_\phi\) of the combined structural equilibrium model \eqref{swigSCM} is \(\phi(\mathcal{G}; C)\).
\end{proposition}
\begin{proof}
     Given an observational structural equilibrium model with the corresponding graph \(\mathcal{G}\), consider the combined structural equilibrium model in \eqref{swigSCM}. 
     
     We first show that the nodes coincide. Consider the set of nodes \(i\not\in C\) such that \(C\cap \textnormal{ant}(i)=\emptyset\) in \(\mathcal{G}\), call this set \(A\), which is the set in Step 2 of Algorithm \ref{newswagalg}. For \(i\in A\), it holds that  \(X_{\tau(i)}=X^{\textnormal{do}(C)}_{\tau(i)}\)  and the arguments (in \(A\)) of the functional assignments of \(X^{\text{do}(C)}_\tau\) in \eqref{swigSCM} can be substituted as
     \begin{align}\label{reex3}
    X^{\text{do}(C)}_\tau&=f^{\text{do}(C)}_\tau(X^{\text{do}(C)}_{\text{pa}(\tau)\backslash A}, X_{\text{pa}(\tau)\cap A}, \epsilon_\tau), \quad  \text{if } \tau\cap C\neq \emptyset\nonumber \text{ and }\\
        X^{\text{do}(C)}_\tau&=f_\tau(X^{\text{do}(C)}_{\text{pa}(\tau)\backslash A}, X_{\text{pa}(\tau)\cap A}, \epsilon_\tau),  \quad \text{if } \tau\cap C=\emptyset.
    \end{align}

 Thus, \eqref{swigSCM} has exactly the same variables as the nodes in \(\phi(\mathcal{G}; C)\). 

 We show that the undirected and directed edges coincide. 
 From \eqref{reex3}, since the law of the functional assignments remains the same as the intervened structural equilibrium model \eqref{intscm}, applying 
 Steps 1 and 2 of Algorithm \ref{obsalg} to the combined structural equilibrium model results in, for nodes \(i^{\textnormal{do}(C)}\), the same \(\textnormal{ne}_{\textnormal{do}}(i^{\textnormal{do}(C)})\) and \(\textnormal{pa}_{\textnormal{do}}(i^{\textnormal{do}(C)})\) as the relabelled \(\textnormal{do}_C(\mathcal{G})\) (by similar argument of Proposition \ref{intrel}), 
 except with nodes \(j^{\textnormal{do}(C)}\in \textnormal{pa}_{\textnormal{do}}(i^{\textnormal{do}(C)})\), where \(j\in A\), being replaced with \(j\). Thus, applying Steps 1 and 2 of Algorithm \ref{obsalg} to \eqref{swigSCM} results in the same undirected and directed edges as after applying Steps 1 and 2 of Algorithm \ref{newswagalg}.

We next show that the bidirected edges coincide. The dependence of the errors between two observational variables and two intervened variables is the same as the observational structural equilibrium model and the intervened structural equilibrium model \eqref{intscm} respectively, thus when applying Algorithm \ref{obsalg}, the added bidirected edges coincide with those of \(\mathcal{G}\) and \(\textnormal{do}_C(\mathcal{G})\) respectively. The error \(\epsilon_\tau\) for \(X_\tau\) and \(X^{\textnormal{do}(C)}_{\tau}\) is shared in \eqref{swigSCM}. The resulting dependence of errors between observational and intervened variables, when used to create bidirected edges using Algorithm \ref{obsalg}, is accounted for in Table \ref{easyalg} of Step 3 of Algorithm \ref{newswagalg}.
\end{proof}

\begin{proposition}[The single world graph \(\phi(\mathcal{G};C)\) is a chain-connected anterial graph]\label{easyantoutput}
    If \(\mathcal{G}\) is a chain-connected anterial graph, then for any node subset \(C\subseteq V\), we have that \(\phi(\mathcal{G};C)\) is a chain-connected anterial graph.
\end{proposition}
\begin{proof}[Proof of Proposition \ref{easyantoutput}]
 
 Note that nodes in \(\phi(\mathcal{G}; C)\) are either nodes from \(\mathcal{G}\) or from \(\textnormal{do}_C(\mathcal{G})\), which we will differentiate with a superscript \(\cdot^{\textnormal{do}(C)}\) (with \(i^{\textnormal{do}(C)}\) being a node from \(\textnormal{do}_C(\mathcal{G})\) that corresponds with the node  \(i\) from \(\mathcal{G}\)). 

 We have the following observations.
\begin{enumerate}
\item 
\label{pp-obvs1}
By Lemma \ref{do is chain}, both \(\mathcal{G}\) and \(\textnormal{do}_C(\mathcal{G})\) are chain-connected anterial graphs.
\item 
\label{pp-obvs2}
For nodes \(i,j\not\in C\), an edge between \(i\) and \(j\) in \(\mathcal{G}\) correspond to the same edge between  \(i^{\textnormal{do}(C)}\) and \(j^{\textnormal{do}(C)}\) in \(\textnormal{do}_C(\mathcal{G})\).
\item 
\label{pp-obvs3}
From Step 2, edges between nodes \(i\) and \(j^{\textnormal{do}(C)}\) in \(\phi(\mathcal{G}; C)\) are either \(i\rightarrow j^{\textnormal{do}(C)}\) or \(i\leftrightarrow j^{\textnormal{do}(C)}\). See, for example, the middle of Figure \ref{newswagfig}. This is because \(i, j\not\in C\) and \(j\not\in \textnormal{ant}(i)\), otherwise \(j\in \textnormal{ant}(i)\) and \(C\cap \textnormal{ant}(j)\neq \emptyset\) (by the condition in Step 2) causing \(C\cap \textnormal{ant}(i)\neq \emptyset\), contradicting the condition in Step 2. 
\item 
\label{pp-obvs4}
An edge between nodes \(i\) and \(j\) in \(\phi(\mathcal{G}; C)\) corresponds to the same edge in \(\mathcal{G}\). If \(i^{\textnormal{do}(C)}\) and \(j^{\textnormal{do}(C)}\) are nodes in \(\phi(\mathcal{G}; C)\),
then
this correspondence of edges between \(\phi(\mathcal{G}; C)\) and \(\textnormal{do}_C(\mathcal{G})\) also holds.

\end{enumerate}

We will show that there are no semi-directed cycles in \(\phi(\mathcal{G}; C)\). By Observations \eqref{pp-obvs1} and \eqref{pp-obvs4}, we only have to show semi-directed cycles consisting of both nodes of the form \(i\) and \(j^{\textnormal{do}(C)}\) 
do not exist in 
\(\phi(\mathcal{G}; C)\). This is implied by \eqref{pp-obvs3}.

We will show that there are no semi-directed paths connecting endpoints of a bidirected edge in \(\phi(\mathcal{G}; C)\). Consider edges of the form \(i\leftrightarrow j\) or \(i^{\textnormal{do}(C)}\leftrightarrow j^{\textnormal{do}(C)}\) in \(\phi(\mathcal{G}; C)\). By 
\eqref{pp-obvs3},
a semi-directed path between \(i\) and \(j\) or \(i^{\textnormal{do}(C)}\) and \(j^{\textnormal{do}(C)}\) in \(\phi(\mathcal{G}; C)\) must only consists of nodes of the form \(k\) or \(k^{\textnormal{do}(C)}\) respectively. 
By 
\eqref{pp-obvs4},
such a semi-directed path connects the endpoints of \(i\leftrightarrow j\) or \(i^{\textnormal{do}(C)}\leftrightarrow j^{\textnormal{do}(C)}\) in \(\mathcal{G}\) or \(\textnormal{do}_C(\mathcal{G})\), respectively, and thus cannot exist by 
\eqref{pp-obvs1}.
For the edge \(i\leftrightarrow j^{\textnormal{do}(C)}\), by 
\eqref{pp-obvs3},
a semi-directed path from \(i\) to \(j^{\textnormal{do}(C)}\) cannot contain nodes of the form \(k^{\textnormal{do}(C)}\) where \(k\in C\) (since there are no undirected or directed edges pointing into \(k^{\textnormal{do}(C)}\) in \(\textnormal{do}_C(\mathcal{G})\) and by 
\eqref{pp-obvs4}
in \(\phi(\mathcal{G}; C)\) as well), 
thus by 
\eqref{pp-obvs2}
would imply a semi-directed path from endpoints of \(i\leftrightarrow j\) in \(\mathcal{G}\).
    
Chain-connectedness of \(\phi(\mathcal{G}; C)\) follows trivially from item 3 of Table \ref{easyalg}, along with \(\mathcal{G}\) and \(\textnormal{do}_C(\mathcal{G})\) being chain-connected 
by \eqref{pp-obvs1}.
\end{proof}

We will show Theorem \ref{correctnewswag} using Theorem \ref{correctgenant} and Proposition \ref{intrel}.

\begin{proof}[Proof of Theorem \ref{correctnewswag}]
By Proposition \ref{swigrel}, we have that \(\phi(\mathcal{G}; C)\) is the corresponding graph of the intervened structural equilibrium model \eqref{swigSCM}. By Proposition \ref{easyantoutput}, we have that \(\phi(\mathcal{G}; C)\) is a chain-connected anterial graph. An application  of  Theorem \ref{correctgenant}  completes the proof.
\end{proof}

\begin{proof}[Proof of Proposition \ref{propswigcg}]
    By Theorem \ref{correctnewswag}, since \(P^*\) is Markovian to \(\phi(\mathcal{G}; C)\), the marginal \(P^*_{\textnormal{sub}}\) is Markovian to \(\phi'(G;C)=\alpha_{\textnormal{m}}(\phi(\mathcal{G}; C); \textnormal{po}(C))\). 
\end{proof}

\subsection{Proof of Theorem \ref{CCSworks}}

To 
prove
Theorem \ref{CCSworks}, we will use the following results from \citet[Theorem 1]{margincond} which shows how \(\alpha_{\textnormal{m}}\) composes, as a function,  with itself.
\begin{proposition}[\citep{margincond}]\label{compose}
    Consider the graph \(\mathcal{G}\) over the nodes \(V\). Given disjoint subsets \(M_1, M_2\subseteq V\), we have that \(\alpha_{\textnormal{m}}(\alpha_{\textnormal{m}}(\mathcal{G}; M_1); M_2)=\alpha_{\textnormal{m}}(\mathcal{G}; M_1\cup M_2) \).
\end{proposition}
%
%
%
\begin{lemma}\label{conmar}
    Consider the graph \(\mathcal{G}\) over \(V\) and subsets \(L,U\subseteq V\) such that \(L\subseteq U\). The  node \(i\) not being adjacent to the node \(j\) in \(\text{max}(\alpha_{\textnormal{m}}(\alpha_{\textnormal{c}}(\mathcal{G};L); V\backslash U))\) is equivalent to the existence of a set \(S\subseteq V\) such that \(i\perp_{\mathcal{G}} j\cd S\) and \(L\subseteq S \subseteq U\). 
\end{lemma}
\begin{proof}
    We have the following equivalence. Let \(\mathcal{G}_0\) denote the graph \(\alpha_{\textnormal{m}}(\alpha_{\textnormal{c}}(\mathcal{G};L);V\backslash U)\).
    \begin{align*}
     i \textnormal{ and } j \textnormal{ not adjacent} \text{ in max}(\mathcal{G}_0)  \quad &\iff \quad i\perp_{\text{max}(\mathcal{G}_0)}j\cd A \textnormal{ for some } A\subseteq U\backslash L \nonumber\\
     &\iff \quad i\perp_{\alpha_{\textnormal{m}}(\alpha_{\textnormal{c}}(\mathcal{G};L); V\backslash U)} j\cd A\nonumber\\
     &\iff \quad i\perp_{\alpha_{\textnormal{c}}(\mathcal{G};L)} j\cd A\\
     &\iff \quad i\perp_{\mathcal{G}} j\cd A\cup L, 
     \end{align*}
     where the first equivalence follows by Definition \ref{maxim} of maximal graphs, the second equivalence follows by Markov equivalence of \(\textnormal{max}(\mathcal{G}')\) and \(\mathcal{G}'\) for any graph \(\mathcal{G}'\) and the last two equivalences follow by the definitions of operations \(\alpha_{\textnormal{m}}\) in \eqref{propmarg} (since \(i,j, A \subseteq (V\backslash L)\big\backslash (V\backslash U)\)) and \(\alpha_{\textnormal{c}}\) in \eqref{propcond} (since \(i,j, A\subseteq V\backslash L\)). Taking \(S=A\cup L\), we have \(i\perp_{\mathcal{G}} j\cd S\) and \(L \subseteq S \subseteq U\).
\end{proof}

\begin{proof}[Proof of Theorem \ref{CCSworks}]
    Let the joint distribution of \(X^{\textnormal{do}(C)}_O, X_C, X_U\) and potentially other variables be Markovian to some anterial graph \(\mathcal{G}^*\) (over nodes \(V^*\)) and denote the nodes representing \(X^{\textnormal{do}(C)}_O\) and \(X_C\) as \(O^{\textnormal{do}(C)}\) and \(C\) respectively.
    To select an adjustment set \(S\), 
    it suffices to select a set \(S\) such that \(O^{\textnormal{do}(C)}\perp_{\mathcal{G}^*} C\cd S\).

    Consider the graph \(\mathcal{G}^*_0= \alpha_{\textnormal{m}}(\alpha_{\textnormal{c}}(\mathcal{G}^*;L);V^*\backslash (U\cup C\cup O^{\textnormal{do}(C)}))\) in Algorithm \ref{ccsalg}. Lemma \ref{conmar} guarantees that non-adjacency of \(O^{\textnormal{do}(C)}\) and \(C\) in \(\textnormal{max}(\mathcal{G}^*_0)\) guarantees the existence of an adjustment set \(S\) such that \(L\subseteq S\subseteq U\). 

    Let \(S_n\) denote the candidate output set \(S\) at step \(n\) of Algorithm \ref{ccsalg}. Given the node \(i_n\in S_n\backslash L\) at step \(n\), if \(O^{\textnormal{do}(C)}\) is not adjacent to \(C\) in \(\text{max}(\alpha_{\textnormal{m}}(\mathcal{G}^*_n;i_n))\), then by Proposition \ref{compose}, we have \(\alpha_{\textnormal{m}}(\mathcal{G}^*_n;i_n)=\alpha_{\textnormal{m}}(\alpha_{\textnormal{c}}(\mathcal{G}^*; L);(V^*\backslash (S_n\cup C\cup O^{\textnormal{do}(C)}))\cup i_n)\). Applying Lemma \ref{conmar}, there exists a set \(S\) such that \(O^{\textnormal{do}(C)}\perp _{\mathcal{G}^*}C\cd S\) and \(L\subseteq S\subseteq S_{n}\backslash i_n\subset U\). Thus, at step \(n\)  when some node \(i_n\) can be excluded from \(S_n\), Algorithm \ref{ccsalg} guarantees the existence of an adjustment set \(S\) that satisfies \(L \subseteq S\subseteq U\) and decreases the cardinality of \(S_n\) by one. Consider the first step \(n\) such that there are no nodes that can be excluded from \(S_n\), then by Proposition \ref{compose} and Lemma \ref{conmar}, it holds that \(O^{\textnormal{do}(C)}\not\perp_{\mathcal{G}^*}C \cd S\) for every set \(S\) such that \(L\subseteq S\subseteq S_n\backslash i_n \) for any node \(i_n\in S_n\). However, the existence of a set \(S\subseteq S_n\) satisfying \(O^{\textnormal{do}(C)}\perp_{\mathcal{G}^*}C \cd S\) is ensured in step \(n-1\), thus \(S_n\) must be an adjustment set.


    Let the joint distribution of \(X^{\textnormal{do}(C)}_O, X_C, X_U\) and potentially other variables be faithful to the anterial graph \(\mathcal{G}^*\). Consider again the first step \(n\) such that no nodes can be excluded from \(S_n\), then \(O^{\textnormal{do}(C)}\not\perp_{\mathcal{G}^*}C \cd S\) for every set \(S\) such that \(L\subseteq S\subseteq S_n\backslash i_n \) for any node \(i_n\in S_n\). By faithfulness, every set \(S\) such that \(L\subseteq S\subset S_n\) is not an adjustment set. However the existence of an adjustment set \(S\subseteq S_n\) is guaranteed from step \(n-1\), thus \(S_n\) must be a minimal adjustment set.
    
    If \(O^{\textnormal{do}(C)}\) is  adjacent to \(C\) in \(\text{max}(\mathcal{G}_0^*)\), then by Lemma \ref{conmar}, we have \(O^{\textnormal{do}(C)}\not\perp_{\mathcal{G}^*}C\cd S\) and by faithfulness, \(O^{\textnormal{do}(C)}\cancel{\ci} C\cd S\),  for every set \(S\) such that \(L\subseteq S\subseteq U\), thus an adjustment set \(S\) satisfying \(L\subseteq S \subseteq U\) does not exist.
    %
    %
    %
    %
\end{proof}

\subsection{The anterial assumption and causal interpretability}\label{sec:anterial}
In general, the corresponding graph of a structural equilibrium model is a chain mixed graph (see Definition \ref{def cmg}), to which the induced joint distribution \(P\) may not be Markovian. We can generalise Algorithm \ref{obsalg} to Algorithm \ref{obsalgant}, which, in general, returns a Markovian chain mixed graph that \(P\) is Markovian to.
\begin{algorithm}[h]
   \caption{Modification of Algorithm \ref{obsalg} for chain mixed graphs}\label{obsalgant}
   {\bfseries Input:} Structural equilibrium model inducing a joint distribution \(P\).
   
 {\bfseries Output:} Chain-connected chain mixed graph \(\mathcal{G}\).
\begin{algorithmic}[1]

\item Perform Step 1 of Algorithm \ref{obsalg} with \(P_{\tau}(\cdot \cd x_{\textnormal{ant}(\tau)})\) in place of \(J^{\tau}(\cdot\cd x_{\text{pa}(\tau})\).
    \item Perform Step 2 of Algorithm \ref{obsalg} with \(P_{i}(\cdot \cd x_{(\tau\backslash i)\cup\textnormal{ant}(\tau)})\) in place of \(J^{\tau}_i(\cdot\cd x_{\tau\backslash i};x_{\textnormal{pa}(\tau)})\).
    \item Perform Step 3 of Algorithm \ref{obsalg}.

\item Call the resulting graph \(\mathcal{G}'\). In \(\mathcal{G}'\), for nodes \(i\) and \(j\) such that
\begin{enumerate}
    \item \(i\in \textnormal{ant}(j)\backslash \textnormal{pa}(j)\), and 
    \item there exists a primitive inducing path \(\pi=\langle i', q_1,\ldots
    , q_n, j \rangle\) between nodes \(i'\in \tau(i)\) and \(j\) such that \(i'\rightarrow q_1\).
\end{enumerate}
 if \(i\cancel{\ci} j\cd \textnormal{ant}(i,j)\) holds for \(P\), then create the edge \(i\rightarrow q_1\).
\end{algorithmic}
\end{algorithm}

When the corresponding graph is anterial, by equality \eqref{argue} in the proof of Theorem \ref{correctgenant}, the following equality holds for part \(\tau\) of the structural equilibrium model,
\begin{align*}
    P_{\tau}(\cdot \cd x_{\textnormal{ant}(\tau)})= J^{\tau}(\cdot\cd x_{\text{pa}(\tau)}),
\end{align*}
which also implies that \(i\ci j\cd \textnormal{ant}(i,j)\) holds for nodes \(i\) and \(j\) in Step 4 of Algorithm \ref{obsalgant}. Thus, when the corresponding graph is anterial, Algorithm \ref{obsalgant} reduces to Algorithm \ref{obsalg}. 
\begin{remark}The additional Step 4 in Algorithm \ref{obsalgant} allows the argument in the proof of Theorem \ref{correctgenant} using \(\mathcal{G}_\textnormal{col}\) to be applied to the case when \(i\in \textnormal{ant}(j)\). The rest of the proof of Theorem \ref{correctgenant} follows. \erk
\end{remark}

The additional edges created in Step 4 of Algorithm \ref{obsalgant} depend on the joint distribution \(P\) which depends on the actual coupling, and not just the dependencies of the error variables. Thus, the edges in Step 4 of Algorithm \ref{obsalgant} lack mechanistic causal interpretability.

Proposition \ref{intrel} implies that Algorithm \ref{intalg} obtains the corresponding graph \(\mathcal{G}_{\textnormal{do}}\) of the intervened structural equilibrium model from the observational graph \(\mathcal{G}\), \emph{without} referring to the interventional distribution \(P^{\textnormal{do}(C)}\). This is possible since in the case of anterial graphs, Algorithm \ref{obsalg} constructs the corresponding \(\textnormal{G}_{\textnormal{do}}\) by only referring to the law of functional assignments \(J^{\tau, \textnormal{do}(C)}(\cdot\cd x_{\textnormal{pa}(\tau)})\), which behaves as described in \eqref{keyeqint}, as accounted for by Algorithm \ref{intalg} in a local and mechanistic manner. However, in the case of chain mixed graphs, Algorithm \ref{obsalgant} constructs \(\mathcal{G}_{\textnormal{do}}\) \emph{with} reference to the interventional distribution \(P^{\textnormal{do}(C)}\). It is unclear how to graphically relate \(\mathcal{G}\) and \(\mathcal{G}_\textnormal{do}\), since we would have to account for how conditional independencies in \(P^{\textnormal{do}(C)}\) relate to conditional independencies in \(P\), and this relationship depends on the actual coupling, and not just dependencies of the error variables as in the case of anterial graphs. Contrast the illustration below with Figure \ref{goodantfig}.


\begin{center}
     \begin{tikzpicture}[]

\node (e1) at (0,-1.5) {Markovian chain mixed graph \(\mathcal{G}\)};
\node (e2) at (0,0) {Observational model, \(P\)};
\node (e3) at (7,0) {Intervened model, \(P^{\textnormal{do}(C)}\)};
\node (e4) at (7,-1.5) {Markovian chain mixed graph \(\mathcal{G}_\textnormal{do}\) };

\draw[<-] (e1) -- (e2) node[midway, left] {Algorithm \ref{obsalgant}};
\draw[->] (e2) -- (e3) node[midway, above] { intervene };
\draw[->] (e3) -- (e4) node[midway, right] {Algorithm \ref{obsalgant} };
\draw[-, draw=none] (e1) -- (e4) node[midway] {\(\sim\)} node[midway, above=0.3cm] {?} node[midway, below=0.3cm] {\scriptsize depends on actual coupling of errors};
    \end{tikzpicture}
\end{center}



\subsection{Including all nodes in SWIG using parallel-worlds graphs}
\label{subappsec}

Consider the case of DAGs.
Modifying the parallel-worlds graph approach \citep{multinetwork}, 
we present an extension of the construction in Algorithm \ref{newswagalg} that includes all the nodes in SWIGs, complementing the comparison between SWIGs and counterfactual graphs in \cite{swig}.

Given a DAG \(\mathcal{G}\), it is possible to order all the nodes \(j\) and \(k\) in the treatment set \(C\) such that \(k>j\) if there does not exist a directed path from \(j\) to \(k\) in \(\mathcal{G}\). With respect to such an ordering of \(C=\{x_1, \ldots, x_n\}\), we can construct all the parallel-worlds graphs \(\mathcal{G}, \textnormal{do}_{[1]}(\mathcal{G}), \ldots,  \textnormal{do}_{[n]}(\mathcal{G})\), where \(\textnormal{do}_{[i]}(\mathcal{G})=\textnormal{do}_{\{x_1, \ldots, x_i\}}(\mathcal{G})\) and by convention \(\textnormal{do}_{[0]}(\mathcal{G})=\mathcal{G}\). In each graph \(\textnormal{do}_{[i]}(\mathcal{G})\), nodes that correspond to node \(j\) in \(\mathcal{G}\) are relabelled as \(j^{[i]}\) (with \(j^{[0]}=j\)). 

The nodes in the parallel-worlds graphs are then merged and removed as follows:
\begin{enumerate}
    \item Consider \(\textnormal{de}_i\), the set of descendants of \(x^{[i]}_i\), including \(x^{[i]}_i\), that are not descendants of \(\{x^{[i]}_{i+1}, \ldots, x^{[i]}_n\}\) in \(\textnormal{do}_{[i]}(\mathcal{G})\). Let \(\textnormal{de}_0\) be the set of nodes that are not descendants of \(\{x_{1}, \ldots, x_n\}\) in \(\mathcal{G}\).
    \item Repeatedly merge, for each \(i\in \{0,\ldots, n-1\}\), as follows. Consider the node \(j^{[i]}\in \textnormal{de}_i\). 
    \begin{itemize}
        \item If \(j\not\in C\), then merge nodes \(j^{[k]}\) of graphs \(\textnormal{do}_{[k]}(\mathcal{G})\) for all \(k\in \{i+1, \ldots,  n\}\) into node \(j^{[i]}\) as in Step 2 of Algorithm \ref{newswagalg}.
        \item If \(j=x_\ell \in C\), then merge nodes \(x_\ell^{[k]}\) of graphs \(\textnormal{do}_{[k]}(\mathcal{G})\) for all \(k\in \{i+1, \ldots,  \ell -1\}\) into node \(j^{[i]}\) as in Step 2 of Algorithm \ref{newswagalg}.
    \end{itemize}
    \item Take the subgraph of the nodes in \(\bigcup_{i=0}^{n-1}\textnormal{de}_i\) and the remaining nodes in \(\textnormal{do}_{[n]}(\mathcal{G})\).
\end{enumerate}
Call the resulting subgraph \(\mathcal{G}(C)\). See Figure \ref{fullswigext} for an illustration for this extension of Algorithm \ref{newswagalg} in the case of DAGs.

\begin{figure}[h]
\begin{center}
\centerline{\includegraphics[width=\columnwidth]{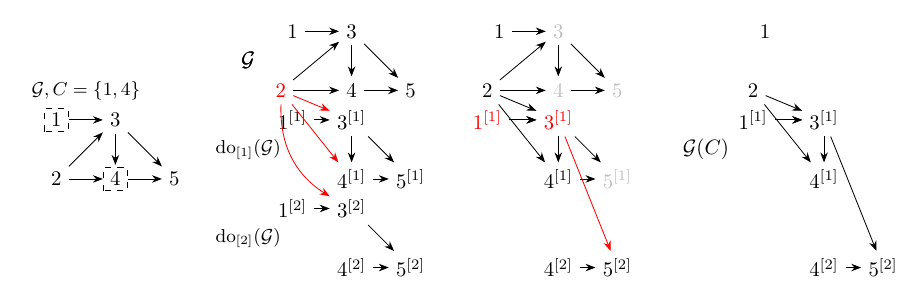}}
\caption{Left: DAG \(\mathcal{G}\) with \(C=\{1,4\}\). Middle 2 graphs: Sequentially merging nodes from different graphs, as in step 2 of Algorithm \ref{newswagalg}. With \(\textnormal{de}_0 =\{1, 2\}\) and \(\textnormal{de}_1 =\{1^{[1]}, 3^{[1]}, 4^{[1]}\}\), the nodes \(2^{[1]}, 2^{[2]}\) are merged into \(2\) and the nodes \(1^{[2]}, 3^{[2]}\) are merged into \(1^{[1]}, 3^{[1]}\) respectively. Right: Output \(\mathcal{G}(C)\) after taking the subgraph over nodes in \(\textnormal{de}_0\cup \textnormal{de}_1\) and the remaining nodes in \(\textnormal{do}_{[2]}(\mathcal{G})\).}
\label{fullswigext}
\end{center}
\end{figure}
We then have the following relating this construction to single-world intervention graphs \citep{swig}.

\begin{proposition}\label{margswif}
    Given a DAG \(\mathcal{G}\) and an ordered treatment set \(C=\{x_1,\ldots, x_n\}\), the output \(\mathcal{G}(C)\) is the single-world intervention graph.
\end{proposition}

\begin{proof}

    Observe that in a single-world intervention graph, after splitting the node \(x_i\) in \(C\), nodes \(j\) that are descendants of \(x_i\), including \(x_i\), but are not descendants of \(\{x_{i+1}, \ldots, x_n\}\) are relabelled as \(j^{[i]}\). These coincide with the nodes in \(\textnormal{de}_i\). 
    While nodes that are not relabelled coincide 
    with the nodes in \(\textnormal{de}_0\) from \(\mathcal{G}\). Thus, since \(\mathcal{G}(C)\) is a subgraph over \(\bigcup_{i=0}^{n-1}\textnormal{de}_i\) and the remaining nodes in \(\textnormal{do}_{[n]}(\mathcal{G})\), the graph \(\mathcal{G}(C)\) contains exactly all the nodes in the single-world intervention graph.

    In the single-world intervention graph, the parent \(k^{[m]}\) (for \(m\leq  i\)) of the node \(j^{[i]}\) corresponds to the parents \(k^{[i]}\) of node \(j^{[i]}\) in \(\textnormal{do}_{[i]}(\mathcal{G})\). Since \(k^{[i]}\) is merged with \(k^{[m]}\) in \(\mathcal{G}(C)\),  the parental relationships in \(\mathcal{G}(C)\) are the same with that in the single-world intervention graph.\qedhere
    
\end{proof}

Repeating the same argument for Proposition \ref{propswigcg}, we see that taking the subgraph over \(\bigcup_{i=0}^{n-1}\textnormal{de}_i\) and the remaining nodes in \(\textnormal{do}_{[n]}(\mathcal{G})\) amounts to marginalising the merged graph. Thus, Propopostion \ref{margswif} formally relates counterfactual graphs \citep{cfgraph} and single-world intervention graphs \citep{swig} by marginalisation.

\newpage

\subsection{Simulation Results}\label{sec:simres}
\begin{figure}[h]
\centerline{\includegraphics[width=0.85\linewidth]{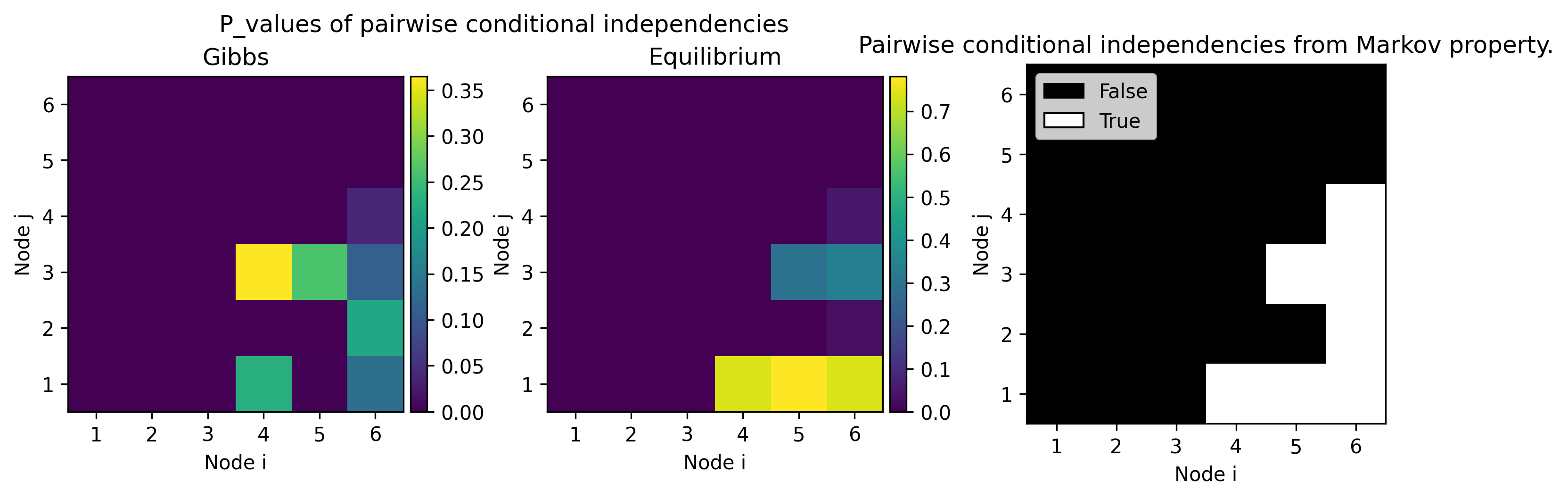}}
\caption{P-values from hypothesis testing of pairwise conditional independencies of the form \(i\ci j\cd \textnormal{ant}(i,j)\). The data is sampled from the structural equilibrium model in Figure \ref{simsetup1} using, Left: Gibbs sampling. Middle: the equilibrium distributions \(J^\tau(\cdot \cd x_{\textnormal{pa}(\tau)})\). Right: Pairwise conditional independencies of the form \(i\ci j\cd \textnormal{ant}(i,j)\) if Markov property w.r.t.\ the corresponding \(\mathcal{G}_1\) holds. }
\label{g1res}
\end{figure}
 \begin{figure}[h]
\centerline{\includegraphics[width=0.85\linewidth]{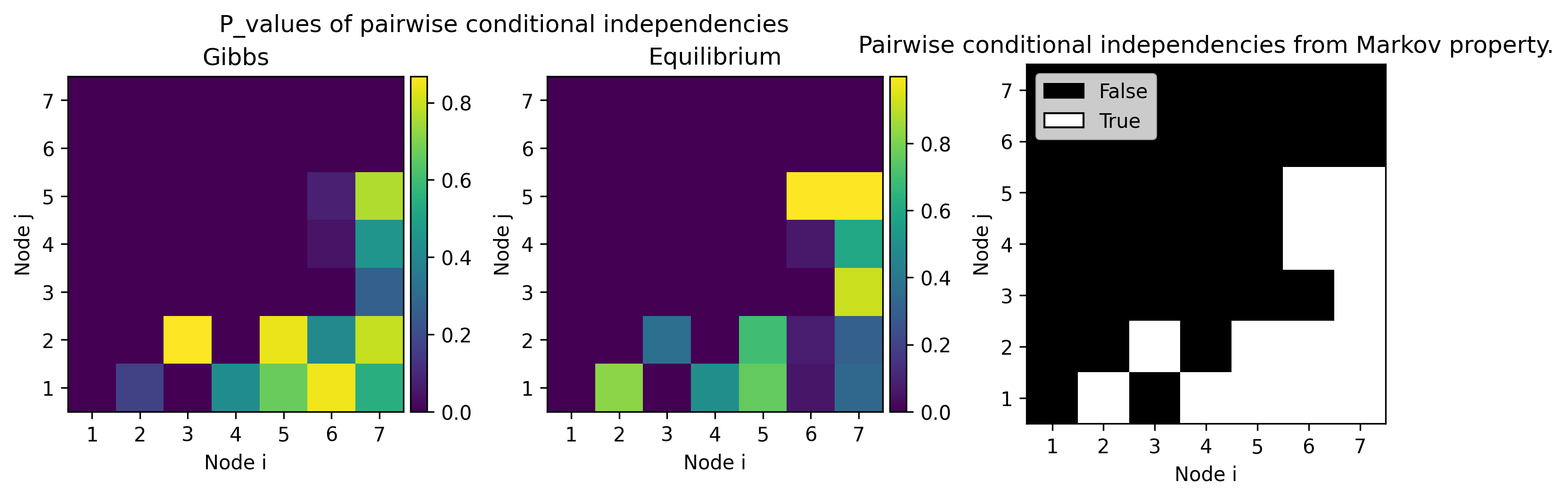}}
\caption{P-values from hypothesis testing of pairwise conditional independencies of the form \(i\ci j\cd \textnormal{ant}(i,j)\). The data is sampled from the structural equilibrium model in Figure \ref{simsetup2} using, Left: Gibbs sampling. Middle: the equilibrium distributions \(J^\tau(\cdot \cd x_{\textnormal{pa}(\tau)})\). Right: Pairwise conditional independencies of the form \(i\ci j\cd \textnormal{ant}(i,j)\) if Markov property w.r.t.\ the corresponding \(\mathcal{G}_2\) holds. }
\label{g2res}
\end{figure}
 \begin{figure}[H]
\centerline{\includegraphics[width=0.85\linewidth]{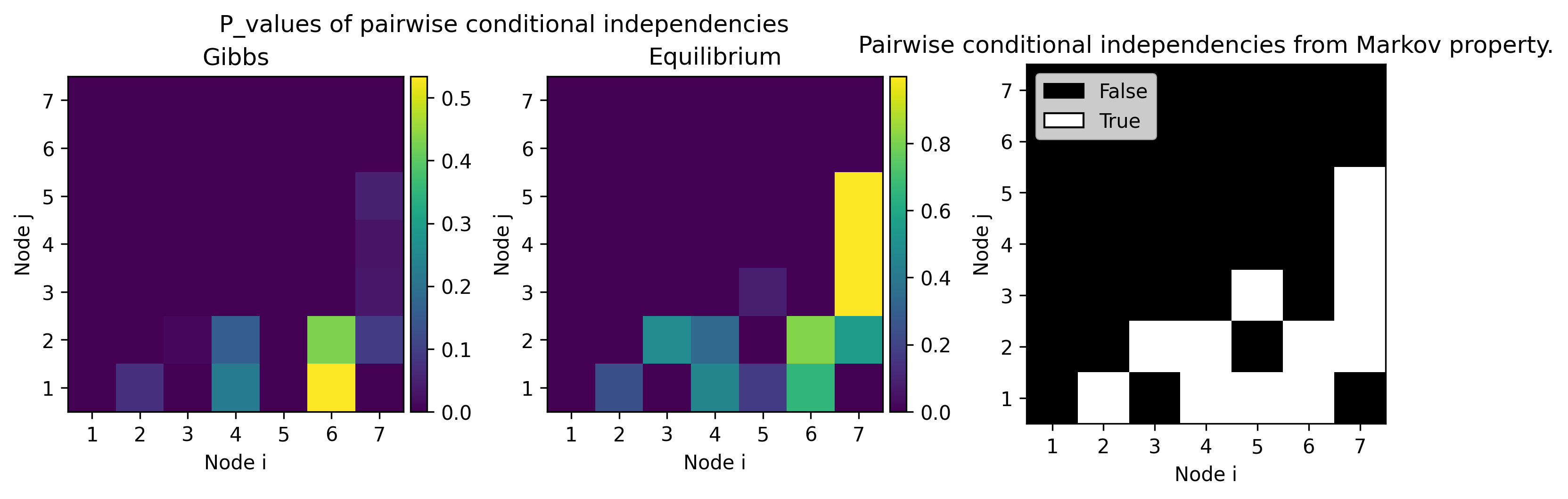}}
\caption{P-values from hypothesis testing of pairwise conditional independencies of the form \(i\ci j\cd \textnormal{ant}(i,j)\). The data is sampled from the structural equilibrium model in Figure \ref{simsetup3} using, Left: Gibbs sampling. Middle: the equilibrium distributions \(J^\tau(\cdot \cd x_{\textnormal{pa}(\tau)})\). Right: Pairwise conditional independencies of the form \(i\ci j\cd \textnormal{ant}(i,j)\) if Markov property w.r.t. the corresponding \(\mathcal{G}_3\) holds. }
\label{g3res}
\end{figure}
\end{appendices}

\bibliographystyle{abbrvnat}
\bibliography{oup-authoring-template}

\end{document}